\documentclass[11pt]{article}
\usepackage[usenames,dvipsnames,svgnames,table]{xcolor} 
\usepackage{kotex}
\usepackage{graphicx}
\usepackage[normalem]{ulem}
\usepackage{amssymb, amsmath, mathtools}
\usepackage{verbatim}
\usepackage[round]{natbib}
\usepackage[colorlinks=true,allcolors=black]{hyperref}
\usepackage[affil-it]{authblk}
\usepackage{mathrsfs}
\usepackage{here}
\usepackage{makecell}
\usepackage{tikz}
\usepackage[shortlabels]{enumitem}
\usepackage{yfonts}
\usepackage{comment}
\usepackage{longtable}
\usepackage{bm, color, amsthm}
\usetikzlibrary{circuits.logic.US,circuits.logic.IEC,fit}
\usepackage{multirow}
\usepackage{float} 

\usepackage{bibunits}
\defaultbibliographystyle{abbrvnat} 
\defaultbibliography{refs}

\usepackage{aliascnt}

\usepackage[noabbrev,capitalise]{cleveref}
\creflabelformat{equation}{#2\textup{#1}#3}

\usepackage{color, colortbl}
\usepackage{booktabs}
\usepackage{geometry}
\usepackage{array}
\usepackage{caption}
\usepackage{subcaption}
\usepackage{dsfont}
\definecolor{Yellow}{rgb}{1,1,0.6}
\usepackage{setspace}
\usepackage{microtype}   

\usepackage{algorithm}
\usepackage{algpseudocode}
\usepackage{adjustbox}
\usepackage{lineno} 

\algrenewcommand\algorithmicrequire{\textbf{Input:}}
\algrenewcommand\algorithmicensure{\textbf{Output:}}

\hypersetup{
 colorlinks,
 citecolor=Blue,
 linkcolor=Blue,
 urlcolor=Blue}
 
\newtheorem{theorem}{Theorem}[section]

\newaliascnt{lemma}{theorem}
\newtheorem{lemma}[lemma]{Lemma}
\aliascntresetthe{lemma}

\newaliascnt{definition}{theorem}
\newtheorem{definition}[definition]{Definition}
\aliascntresetthe{definition}

\newaliascnt{proposition}{theorem}
\newtheorem{proposition}[proposition]{Proposition}
\aliascntresetthe{proposition}

\newaliascnt{condition}{theorem}
\newtheorem{condition}[condition]{Condition}
\aliascntresetthe{condition}

\newaliascnt{corollary}{theorem}
\newtheorem{corollary}[corollary]{Corollary}
\aliascntresetthe{corollary}

\newaliascnt{result}{theorem}

\aliascntresetthe{result}

\newaliascnt{remark}{theorem}

\aliascntresetthe{remark}

\newaliascnt{remarks}{theorem}

\aliascntresetthe{remarks}

\crefname{theorem}{Theorem}{Theorems}
\Crefname{theorem}{Theorem}{Theorems}

\crefname{lemma}{Lemma}{Lemmas}
\Crefname{lemma}{Lemma}{Lemmas}

\crefname{definition}{Definition}{Definitions}
\Crefname{definition}{Definition}{Definitions}

\crefname{proposition}{Proposition}{Propositions}
\Crefname{proposition}{Proposition}{Propositions}

\crefname{condition}{Condition}{Conditions}
\Crefname{condition}{Condition}{Conditions}

\crefname{corollary}{Corollary}{Corollaries}
\Crefname{corollary}{Corollary}{Corollaries}

\crefname{result}{Result}{Results}
\Crefname{result}{Result}{Results}

\crefname{remark}{Remark}{Remarks}
\Crefname{remark}{Remark}{Remarks}

\crefname{remarks}{Remarks}{Remarks}
\Crefname{remarks}{Remarks}{Remarks}

\DeclareMathOperator*{\argmin}{argmin}
\DeclareMathOperator*{\argmax}{argmax}
\newcommand{\E}{\mathbb{E}}

\newcommand{\V}{{\rm\mathbb{V}ar}}
\newcommand{\C}{{\rm\mathbb{C}ov}}

\newcommand{\bea}{\begin{eqnarray*}}
	\newcommand{\eea}{\end{eqnarray*}}
\newcommand{\bean}{\begin{eqnarray}}
	\newcommand{\eean}{\end{eqnarray}}
\newcommand{\benu}{\begin{enumerate}}
	\newcommand{\eenu}{\end{enumerate}}

\newcommand{\bbR}{\mathbb{R}}
\newcommand{\bbN}{\mathbb{N}}

\newcommand{\cM}{\mathcal{M}}
\newcommand{\cN}{\mathcal{N}}

\newcommand{\cX}{\mathcal{X}}

\newcommand{\cA}{\mathcal{A}}

\providecommand{\keywords}[1]{\textbf{\textit{Keywords---}} #1}

\newcommand{\jc}[1]{\textcolor{magenta}{#1}}

\title{Robust Bayesian Predictive Model Selection using \\ Bregman Divergence}
\author[1]{Jongwoo Choi\footnote{Corresponding author, E-mail: cjw7779@gmail.com}}
\author[2]{Neil A. Spencer}
\author[2]{Dipak K. Dey}
\affil[1]{Department of Biostatistics, Harvard University}
\affil[2]{Department of Statistics, University of Connecticut}

\date{} 

\begin{document}
\begin{bibunit}

\maketitle
\vspace{-1em}

\begin{abstract}
Predictive Bayesian model comparison often relies on leave-one-out (LOO) cross-validation criteria such as the expected log predictive density (ELPD). 
However, model rankings can be overly sensitive to outliers and tail mismatch because ELPD is based on the log score. 
We propose a score-matched generalized ELPD framework that replaces the log score by a Bregman scoring rule to update model parameters through a generalized posterior and to evaluate LOO predictive utility.
Candidate posterior predictive distributions are ranked by out-of-sample utility under the chosen scoring rule, yielding a direct proper-score generalization of standard ELPD.
We focus especially on the $\beta$-divergence family, where $\beta$ controls the sensitivity of predictive comparison to low-density observations.
Under model misspecification, the procedure asymptotically selects the model whose predictive distribution is closest to the data-generating process under the chosen Bregman divergence. 
A simulation study and applications to microbial and forensic data show that the generalized ELPD can change the selected model through reduced sensitivity to low-density observations.
\end{abstract}




\section{Introduction}\label{sec:intro}
Model selection is ubiquitous in statistical analysis. When the goal is prediction, a natural way to compare competing models is to consider how well they forecast new, unseen observations from the same underlying population \citep{Lindley1968, vehtari2012survey}. This task is known as predictive model comparison, and it ranks models by their expected out-of-sample predictive performance. A widely used approach to predictive comparison is leave-one-out (LOO) cross-validation, which estimates each model's expected log predictive density (ELPD; \citealp{vehtari2017practical}). These tools evaluate models through their implied predictive distribution, targeting the model with the best predictive performance---even if none captures the true data-generating process (DGP) exactly. This aligns with the $\mathcal{M}$-open regime \citep{bernardo1994bayesian}.

Despite their broad success, ELPD-based criteria can behave undesirably in the presence of outliers or heavy-tailed observations. To make this concrete, consider the contaminated normal DGP $q(x) = (1-\epsilon) \, \cN(x; \,0, 1) + \epsilon \, \cN(x; \, 0, 10^2)$, where $\epsilon = 0.2$ represents a $20\%$ contamination component with much larger variance. Consider two candidate models: $M_1: \cN(0,1)$, which captures the central bulk, and $M_2: t_2(3,1)$. When models are compared using ELPD-LOO, a small fraction of extreme observations can dominate the comparison. In this example, sufficiently extreme draws from the outlying component can pull ELPD toward $M_2$, even though $M_2$ fits the bulk near $0$ poorly. Thus, a model can win largely by protecting itself against rare extremes, at the expense of representing the central patterns that drive most predictions. 


\begin{figure}
    \centering
    \includegraphics[width=0.6\linewidth]{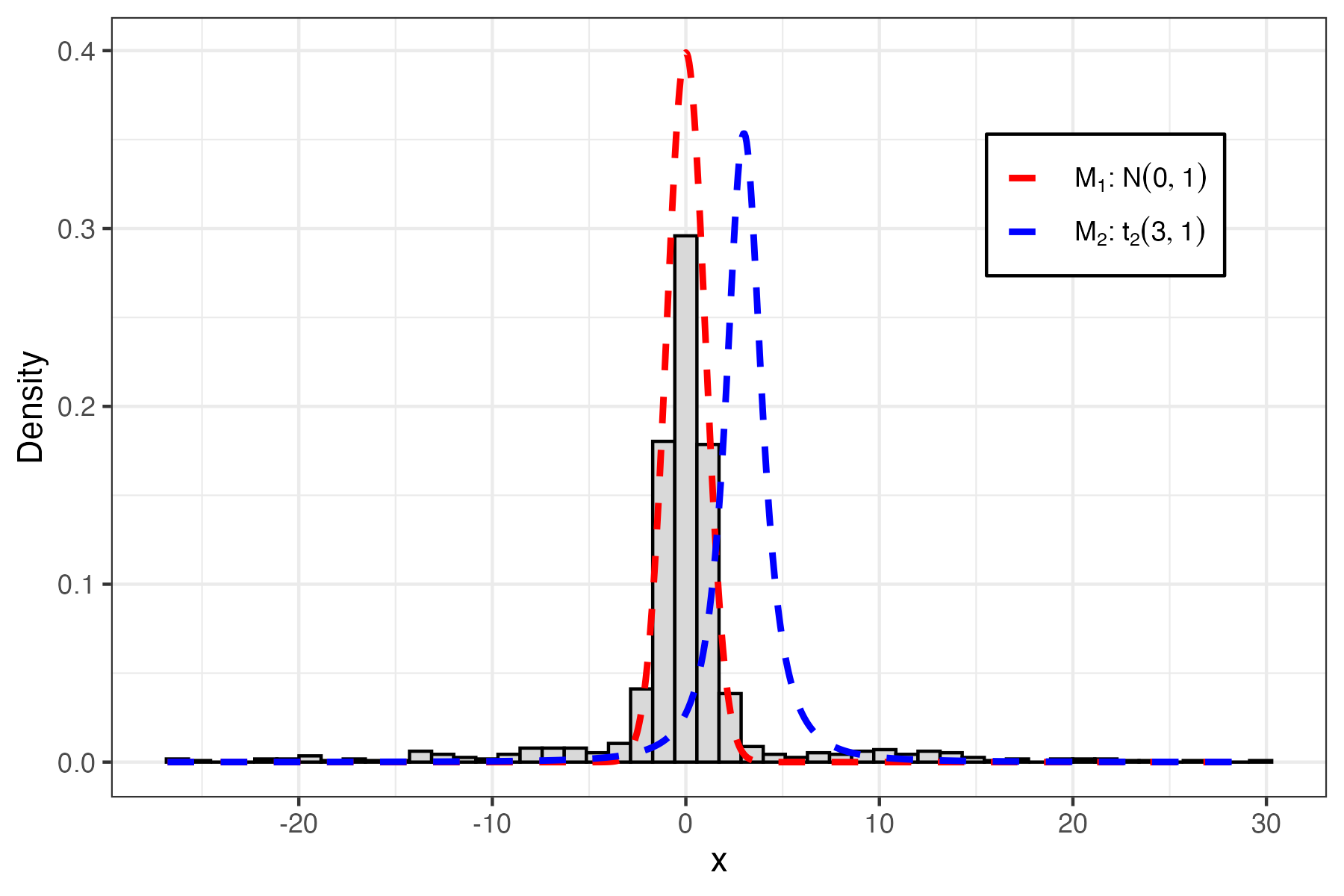}
    \captionsetup{font=small}
    \caption{Histogram shows $n=1000$ simulated data from $q(x) = (1-\epsilon) \, \cN(x; \,0, 1) + \epsilon \, \cN(x; \, 0, 10^2)$, with overlays of $M_1: \cN(0, 1)$ (center-correct, light-tailed) and $M_2: t_2(3, 1) $ (miscentered, heavy-tailed). }
    \label{fig:fig_toy}
\end{figure}

This example reflects a broader limitation of log-predictive-density-based comparison. In expectation under the DGP, ELPD targets the model whose predictive distribution provides the closest approximation to the DGP in Kullback--Leibler (KL) divergence \citep{sawa1978information, vehtari2012survey}. Since the KL divergence is highly sensitive to regions where a predictive density assigns extremely small probability, model comparison can be overly influenced by outliers or tail mismatch \citep{jewson2018principles}. This motivates predictive criteria that reduce sensitivity to tail observations while remaining decision-theoretic in the $\mathcal{M}$-open regime.

Generalized Bayes updating \citep{bissiri2016general} provides a principled way to update prior beliefs using losses rather than likelihoods, and scoring-rule-based methods \citep{gneiting2007strictly} have been developed to address robustness, model misspecification, and 
intractable-likelihood settings \citep{giummole2019objective, knoblauch2022optimization, matsubara2022robust, pacchiardi2024generalized}. Recent work on stability 
suggests that posterior predictive distributions, rather than parameter posteriors, are often the primary objects of interest for robust inference and decision making \citep{jewson2024stability}. From this predictive perspective, model comparison and combination methods evaluate posterior predictive distributions through their out-of-sample utility, typically using proper scoring rules and LOO approximations \citep{piironen2017comparison, yao2018using}.

These developments suggest a natural question for robust predictive model selection: can we replace the log-score target in LOO predictive comparison by a non-log scoring rule, while using the same criterion to update parameters, construct posterior predictives, and evaluate out-of-sample predictive utility? Such a framework should preserve the $\cM$-open predictive interpretation of ELPD, but allow us to adjust sensitivity to outliers and tail mismatch. This paper develops such a generalized predictive comparison framework and studies its robustness, asymptotic target, and practical computation.

\subsection{Contributions of this paper}

We pursue the goal above by replacing the KL target implicit in standard predictive criteria with a broader class of discrepancies. Concretely, we propose a predictive model selection framework based on Bregman divergences (BD; \citealp{bregman1967relaxation}), which include KL divergence as a special case and are naturally linked to proper scoring rules \citep{grunwald2004game, gneiting2007strictly}. The scoring rule provides a decision-theoretic utility for evaluating predictive distributions and also yields a principled loss for generalized Bayes updating \citep{pacchiardi2024generalized}.

This creates a natural predictive workflow. For each candidate model, we update the prior over its parameters using the Bregman score induced by the chosen BD, form the posterior predictive density, and evaluate that density by LOO under the same score. This yields a generalized predictive accuracy criterion, which we call generalized ELPD (g-ELPD) recovering standard ELPD as a special case. 


A key advantage of the Bregman formulation is that the generalized posterior predictive can be obtained by simply marginalizing over the generalized parameter posterior, exactly as in standard Bayesian machinery, while remaining the Bayes action under the same divergence used for updating and evaluation. This tractability is not a general feature of arbitrary divergence-based losses. 

Throughout, we emphasize the $\beta$-divergence \citep{basu1998robust}, also known as density power divergence \citep{ghosh2016robust}, as a practically useful subclass of Bregman divergences that yields robustness under contamination and misspecification. This divergence has been increasingly used in recent years, with applications including model criticism and robust parameter inference \citep{goh2014bayesian, knoblauch2018doubly, knoblauch2022optimization, girardi2020robust, sugasawa2020robust, jewson2024stability}. In our framework, $\beta$ acts as an outlier-sensitivity parameter: $\beta=1$ recovers the usual log-score comparison, while moderate $\beta>1$ 
reduces the leverage of observations assigned very small predictive density.

Our contributions are summarized as follows.
\begin{enumerate}
    \item We formulate Bayesian predictive model comparison as a score-matched proper-scoring-rule decision problem. For each candidate model, the same Bregman score is used for generalized Bayes updating, posterior predictive construction, and LOO predictive evaluation. 
    We also show that for Bregman scores, the ordinary posterior predictive distribution is the Bregman Bayes action for the posterior distribution over model-implied predictive densities.

    \item Under misspecification, we establish posterior concentration, posterior predictive consistency, and model selection consistency under bounded-score regularity conditions. The limiting selected model is the candidate whose pseudo-true predictive density minimizes the chosen Bregman divergence to the data-generating process. We also characterize the limiting target obtained when the updating and evaluation scores are not matched.
    
    \item For the $\beta$-divergence family, we formalize the robustness mechanism through bounded pairwise score contributions and a contamination-stability bound. These results explain how $\beta$-gELPD limits the observation-level leverage of low predictive-density observations relative to log-score ELPD. We also adapt PSIS-LOO to compute the proposed criterion. 

    

    
\end{enumerate}

\subsection{Related work}


Predictive criteria have long played a central role in Bayesian model comparison, assessing models by their out-of-sample predictive performance \citep{stone1974cross, geisser1975predictive, gelfand1992model, bernardo1994bayesian, vehtari2012survey}. This perspective continues to motivate active research on predictive model comparison, including stacking, uncertainty quantification, and projection-based methods \citep{piironen2017comparison, yao2018using, piironen2020projective, sivula2025uncertainty, mclatchie2025advances}. When the primary goal is prediction, model averaging and stacking can be preferable to selecting a single model \citep{hoeting1999bayesian, yao2018using}. In many applications, however, the decision problem is to select a single model -- for interpretability, for reporting and inference under a single generative model, or formal comparison between competing models. 

Classical robust Bayesian analysis often focuses on sensitivity to modeling choices, especially the prior, and develops methods to assess or control how posterior inferences change under perturbations \citep{berger1994overview}. However, as \citet{berger1985statistical} notes, priors cannot typically be specified without error, particularly when elicitation time or information is limited. Motivated by such concerns and by ideas from frequentist robust statistics \citep{huber1981robust}, more recent work targets robustness to model misspecification and outliers through several complementary strategies: replacing the likelihood by a loss \citep{bissiri2016general}, modifying or tempering the likelihood \citep{grunwald2012safe, Bhattacharya2019fractional, miller2019robust}, and adopting divergence-based remedies \citep{hooker2014bayesian, ghosh2016robust, nakagawa2020robust, matsubara2022robust}.

Within the divergence-based literature, a common approach is to generalize Bayesian updating by substituting the implicit KL divergence target with an alternative divergence (see, for example, \citealp{dey1994robust, grunwald2004game, jewson2018principles}). Prominent robust choices include the $\beta$-divergence \citep{basu1998robust, amari2009alpha, ghosh2016robust}. This line of work primarily provides principled robust alternatives to likelihood-based updating for parameter inference. The closest contribution to ours in this inference direction is \citet{jewson2024stability}, which studies stability properties of $\beta$-divergence generalized Bayes updating, with a particular focus on stability of posterior predictive inferences under perturbations of the likelihood specification and the DGP. In contrast, our focus is on the predictive model comparison across competing models.

Our work is different from existing Bregman-divergence-based model selection approaches \citep{goh2021bayesian}, whose primary targets and guarantees are developed under a more $\cM$-closed regime (i.e., assuming one of the candidate models actually recovers the true DGP) and need not coincide with the $\cM$-open predictive objective considered here. To our knowledge, existing methods do not jointly provide: divergence-based generalized updating, predictive model comparison under the same non-log scoring rule, and an asymptotic theory for the resulting score-matched predictive selection criterion.


The paper is structured as follows: \cref{sec:pre} provides the necessary background tools and \cref{sec:meth} presents our Bregman-geometric predictive model selection framework. \cref{sec:theory} develops the asymptotic theory, including model selection consistency. \cref{sec:comp} presents the PSIS-based computation. \cref{sec:examples} revisits the motivating example in \cref{fig:fig_toy} and then presents two applications, one to thermal performance curve model selection in microbial ecology and one to spatial modeling in forensic footwear analysis. 

\section{Preliminaries}\label{sec:pre}
This section reviews the tools and notation used throughout the paper. \cref{sec:pre:notation} sets notation, \cref{sec:pre:decision} frames predictive model selection as a decision problem -- namely, choosing a predictive distribution to maximize the expected out-of-sample utility under the true DGP. \cref{sec:pre:scoring} introduces proper scoring rules and Bregman scoring rules. \cref{sec:pre:gen_bayes} reviews generalized Bayes updating, which yields score-based posterior distributions over parameters.

\subsection{Notation and setup}\label{sec:pre:notation}
Let $(\cX, \cA)$ be the sample space and let $X_1, X_2, \ldots \in \cX$ be independent and identically distributed random variables with common distribution $G$. When convenient, we write $X$ for a generic draw from $G$. Assume $G$ admits a density $g$ with respect to a $\sigma$-finite measure $\lambda$ on $\cX$, for example, the Lebesgue measure for continuous observations or the counting measure for discrete observations. For each $n \in \bbN$, write $X_{1:n} = (X_1, \ldots, X_n)$ and let $x_{1:n}$ denote an observed realization. Let $\tilde{X}$ denote an unseen observation drawn from $G$, independent of $X_{1:n}$.


Consider a finite set of candidate models $\cM = \{M_1, \ldots, M_K\}$. We do not assume $G$ coincides with one of the candidate models, that is, all $K$ models may be misspecified. Model $M_k$ is specified by a family of densities $\{f_k(\cdot; \theta_k): \theta_k \in \Theta_k\}$ (with respect to $\lambda$) and a prior probability measure $\Pi_k$ on $\Theta_k$. Here, $\Theta_k \subseteq \bbR^{d_k}$ is the parameter space of finite dimension $d_k$. When $\Pi_k$ admits a density with respect to Lebesgue measure on $\Theta_k$, we write $\Pi_k(d\theta_k) = \pi_k(\theta_k)d\theta_k$. Throughout the paper, we present results using densities whenever they exist.

\subsection{Predictive model selection as a decision problem}\label{sec:pre:decision}
Predictive model selection aims to choose, from a finite set of candidate models $\{M_k\}_{k=1}^K$, the model whose predictive distribution performs best for unseen observations drawn from the same DGP. After observing $x_{1:n}$, each candidate model $M_k$ induces a predictive distribution for a new observation $\tilde{X}$; we write its predictive density (with respect to $\lambda$) as $p_k(\cdot \mid x_{1:n})$. Let $u(\cdot,\cdot)$ be a utility function \citep{berger1985statistical} such that $u\big(p_k(\cdot \mid x_{1:n}\big), \tilde{X})$ is integrable under the true DGP $G$. That is, $u(p, \tilde{x})$ provides a measure of predictive performance when we use the predictive density $p$ and the realized outcome is $\tilde{x}$. 

Our target is to select the model whose predictive distribution maximizes expected utility under the true DGP:
\begin{equation}
    \widehat M = \argmax_{k \in \{1,\ldots, K\}} \E \Big( u\big( p_k(\cdot \mid x_{1:n}), \tilde{X} \big) \Big) = \argmax_{k \in \{1,\ldots, K\}} \int_{\cX} u\big(p_k(\cdot \mid x_{1:n}), \tilde{x}\big) g(\tilde{x}) \lambda( d\tilde{x}).
\end{equation}
This objective can be viewed through a decision-theoretic framework \citep{bernardo1994bayesian, vehtari2012survey}, where the actions are the candidate predictive distributions $\big\{p_k(\cdot \mid x_{1:n})\big\}_{k=1}^K$ and $u$ is the associated utility. This framing is particularly natural in the $\cM$-open regime, where none of the models $M_k$ are assumed to be correctly specified: the goal is not to identify a ``true'' model, but rather to compare competing predictive behaviors via out-of-sample utility. In practice, the expected utility is often estimated by leave-one-out cross-validation (LOO-CV; \citealp{geisser1975predictive, 
vehtari2017practical}).

\subsection{Proper Scoring Rules and Bregman Divergence}\label{sec:pre:scoring}
The predictive model comparison problem is determined by the choice of utility function $u(\cdot,\cdot)$ used to evaluate candidate predictive distributions with respect to the DGP. A standard choice in a Bayesian context is to take $u$ to be a (proper) scoring rule \citep{gneiting2007strictly}, which evaluates probabilistic predictions by assigning a numerical score based on the predictive distribution and the realized observation, and can be interpreted as a utility to be maximized in expectation \citep{bernardo1994bayesian}. 

Following the notation of \citet{gneiting2007strictly}, a scoring rule is a function $S: \mathcal{P} \times \cX \to \bbR$, where $\mathcal{P}$ is a collection of probability measures on $\cX$. For any $P \in \mathcal{P}$ and realized outcome $x \in \cX$, the quantity $S(P, x)$ measures the quality of the predictive distribution $P$ when $x$ occurs. When $P$ admits a density $p$ with respect to $\lambda$, we may write $S(p,x)$. For $P, Q \in \mathcal{P}$, define the expected score under $Q$ by 
\begin{equation*}
    S(P,Q) := \int_{\cX} S(P, x) Q(dx) = \int_\cX S(P,x) q(x) \lambda(dx),
\end{equation*}
when $Q$ has density $q$ with respect to $\lambda$. The scoring rule $S$ is proper (relative to $\mathcal{P}$) if for all $P, Q \in \mathcal{P}$, $S(Q,Q) \ge S(P,Q)$ and it is strictly proper if equality holds only when $P = Q$. Thus, under a strictly proper scoring rule, the expected score $S(P,Q)$ is uniquely maximized by $P=Q$. 

In this paper, we compare predictive distributions by maximizing expected score, so we now specialize to a strictly proper family that we will use throughout.
\begin{definition}[Bregman scoring rule; \citet{grunwald2004game}]\label{def:bregman_score}
Let $\phi:[0,\infty) \to \bbR$ be strictly convex and continuously differentiable. Let $P \in \mathcal P$ be absolutely continuous with respect to $\lambda$ with density $p$. The (separable) Bregman scoring rule generated by $\phi$ is
\begin{equation}\label{eq:bregman_score}
    S^\phi (P, x) = \phi'(p(x)) - \int_\cX \Big( \phi'(p(t)) \, p(t) - \phi(p(t)) \Big) \lambda (dt),
\end{equation}
for $x\in\cX$, whenever the integral is finite.
\end{definition}
A commonly used member of this family is obtained by taking $\phi(u) = \tfrac{u^\beta}{\beta(\beta - 1)}$ for $\beta>0$, $\beta\neq 1$. Substituting this choice into \cref{eq:bregman_score} yields the associated $\beta$-score 
\begin{equation}\label{eq:beta_score}
    S^\beta(P,x) = \frac{p(x)^{\beta-1}}{\beta-1} - \frac{1}{\beta}\int_\cX p(t)^\beta\,\lambda(dt),
\end{equation}
whenever the integral is finite. The log score is recovered in the continuous limit $\beta\to 1$, up to an additive constant independent of $x$. Through the choice of $\beta$, one can control sensitivity to discrepancies in low-density regions. When $\beta>1$, the observation-specific term $p(x)^{\beta-1}/(\beta-1)$ remains bounded as $p(x) \to 0$, unlike the log score. Thus, low predictive density observations have reduced influence, making $\beta>1$ useful for robust predictive comparison. Values $\beta<1$ have the opposite behavior and are therefore not the focus of this paper.


Proper scoring rules induce a corresponding discrepancy between predictive distributions. In particular, each proper score defines a notion of ``closeness'' between predictive distributions (via its induced divergence), so maximizing expected predictive utility is equivalent to minimizing this discrepancy to the DGP. Specifically, define the generalized entropy $H(Q) := S(Q, Q) = \E_{X \sim Q}\big(S(Q, X)\big)$, and define the associated divergence as $D(P, Q) := H(Q) - S(P, Q)$. If $S$ is proper, then $D(P,Q) \ge 0$ for all $P,Q \in \mathcal{P}$, and if $S$ is strictly proper then $D(P,Q)=0$ if and only if $P=Q$. Hence, for fixed $Q$, maximizing the expected score $S(P,Q)$ over $P$ is equivalent to minimizing $D(P,Q)$, since $H(Q)$ does not depend on $P$. For the Bregman score $S^\phi$, the induced divergence coincides with the separable Bregman divergence \citep{bregman1967relaxation} generated by the same $\phi$:
\begin{equation}\label{eq:bregman_div}
    D^\phi(P, Q) = D^\phi(q \| p) := \int_\cX \Big( \phi(q(x)) - \phi(p(x)) - \phi'(p(x))\big( q(x) - p(x) \big) \Big) \lambda (dx),
\end{equation}
whenever the integral is finite. Consequently, maximizing expected Bregman score is equivalent to minimizing $D^\phi(\cdot\|\cdot)$ to the DGP. Bregman divergences generalize many familiar divergence measures, including KL divergence, squared Euclidean distance, and Itakura–Saito distance \citep{banerjee2005clustering}. This generality is useful because it lets us work within one flexible divergence family, so we can choose a loss that matches the problem (e.g., robustness to outliers) rather than committing to a single default such as KL divergence.

The $\beta$-score in \cref{eq:beta_score} corresponds to the density power divergence \citep{basu1998robust} under the reparameterization:
\begin{equation}\label{eq:beta_div}
    D^\beta(q \| p) = \int_\cX \left( \frac{q(x)^\beta}{\beta(\beta - 1)} + \frac{p(x)^\beta}{\beta} - \frac{q(x) p(x)^{\beta - 1}}{\beta - 1} \right) \lambda(dx),
\end{equation}
with $\beta=1$ case defined by continuity, recovering KL divergence as $\beta \to 1$, and the $\beta \to 0$ limit recovering the Itakura--Saito divergence.

\subsection{Generalized Bayesian updating}\label{sec:pre:gen_bayes}
Recall from \cref{sec:pre:decision} that predictive model selection compares the posterior predictive distributions induced by the candidate models. Therefore, for each candidate model $M_k$, we first need a principled way to update its parameter distribution after observing the data; this is a necessary step on the way to forming posterior predictives for comparison. Consider the setup in \cref{sec:pre:notation}. In each candidate model $M_k$, parameter inference amounts to updating a prior $\pi_k(\theta_k)$ on $\Theta_k$ after observing $x_{1:n}$. Standard Bayesian updating yields the posterior
\begin{equation}
    \pi_{k,n}(\theta_k \mid x_{1:n}) = \frac{ \pi_k(\theta_k) \exp \bigg(\sum_{i=1}^n \log f_k(x_i; \theta_k) \bigg)}{ \int_{\Theta_k}  \pi_k(\vartheta_k) \exp \bigg(\sum_{i=1}^n \log f_k(x_i; \vartheta_k) \bigg) d\vartheta_k }.
\end{equation}
In $\cM$-open regime, the posterior asymptotically concentrates (under certain conditions) within the set of parameter values that maximize the expected log score -- equivalently, that minimize KL divergence from the true DGP to the model family \citep{berk1966limiting, Bunke1998}. 

Since our goal is to choose the model whose predictive distribution minimizes a Bregman divergence to the DGP, it is natural to use a parameter inference strategy that targets parameters optimizing the same induced Bregman score. Thus, we adopt a loss-based updating method known as generalized Bayes updating \citep{bissiri2016general}.
Let $\ell: \cX \times \Theta_k \to \bbR$ be a loss under model $M_k$. The generalized posterior up to a normalizing constant is 
\begin{equation}\label{eq:gen_post_loss}
    \pi^\ell_{k,n}(\theta_k \mid x_{1:n}) \propto \pi_k(\theta_k) \, \exp \bigg( -w \sum_{i=1}^n \ell\big(x_i, \theta_k \big) \bigg),
\end{equation}
where $w >0$ is a temperature parameter related to the information in the data relative to the information in $\pi_k(\theta_k)$. We set $w=1$ throughout to focus on the effect of the scoring rule itself. Temperature calibration is an important issue in generalized Bayes; however, recent work suggests that predictive performance can be relatively insensitive to $w$ over positive compact ranges once posterior concentration occurs \citep{mclatchie2025predictive}. Also, since we use scoring rules as a utility as in \cref{sec:pre:scoring}, we can express generalized Bayes directly in scoring-rule form \citep{giummole2019objective, pacchiardi2024generalized}. 


The updating via \cref{eq:gen_post_loss} is coherent (i.e., Bayesian additivity) in the sense that sequential updating over observations yields the same result as updating once on the joint data \citep{bissiri2016general}. \citet{miller2021asymptotic} establishes the concentration and asymptotic normality results for \cref{eq:gen_post_loss}.


\section{Robust Bayesian predictive model selection using Bregman Divergence}\label{sec:meth}
This section develops the proposed predictive model selection framework. \cref{sec:meth:breg_post} defines the Bregman posterior and the associated posterior predictive density for each fixed model $M_k$, obtained by averaging the model densities $f_k(\cdot; \theta_k)$ over $\theta_k$ under the Bregman posterior. We then show that this predictive density is the Bregman centroid, or Bayes action of those model densities. \cref{sec:meth:gelpd} defines generalized ELPD (g-ELPD) by evaluating this predictive density with the same Bregman score in a leave-one-out scheme. \cref{sec:meth:beta} specializes the framework to the $\beta$-divergence family and derives the robustness properties used for outlier and tail-sensitive model comparison.



\subsection{Bregman posterior predictive as a geometric center}\label{sec:meth:breg_post}


In \cref{sec:pre:gen_bayes} we introduced the generalized Bayes updating as a framework for parameter inference driven by a loss. Here, we take $\ell(x_i, \theta_k) = -S^\phi\big(f_k(\cdot;\theta_k), x_i\big)$, where $S^\phi$ is the Bregman scoring rule in \cref{eq:bregman_score} induced by a fixed strictly convex function $\phi$. Then \cref{eq:gen_post_loss} becomes
\begin{equation}\label{eq:breg_post}
    \pi_{k,n}^{\phi}(\theta_k \mid x_{1:n}) \propto \pi_k(\theta_k) \, \exp \bigg( \sum_{i=1}^n S^{\phi}\big(f_k(\cdot;\theta_k),x_i\big) \bigg).
\end{equation}
We refer to \cref{eq:breg_post} as the Bregman posterior induced by $S^\phi$. When $S^\phi \big(f_k(\cdot;\theta_k),x_i \big)=\log f_k(x_i; \theta_k)$ up to an additive term not depending on $\theta_k$, \cref{eq:breg_post} recovers the standard Bayesian posterior.

The corresponding Bregman posterior predictive density under model $M_k$ is:
\begin{equation}\label{eq:breg_pred}
    p_{k,n}^\phi(\tilde x \mid x_{1:n}) := \int_{\Theta_k} f_k(\tilde{x}; \theta_k) \, \pi_{k,n}^{\phi}(\theta_k \mid x_{1:n}) \, d\theta_k.
\end{equation}
The Bregman information identity in \cref{prop:bd_identity} applies directly here. In \cref{eq:breg_pred}, $p_{k,n}^\phi(\tilde x \mid x_{1:n})$ is the posterior average of $f_k(\cdot; \theta_k)$ with respect to the Bregman posterior $\pi_{k,n}^{\phi}(\theta_k \mid x_{1:n})$. Hence it is not only the usual posterior predictive mixture; it is the Bregman centroid of the model densities $f_k(\cdot; \theta_k)$ averaged with respect to $\pi_{k,n}^{\phi}(\theta_k \mid x_{1:n})$, under the same divergence used for updating and evaluation. This connection is precisely what makes the Bregman posterior predictive geometry-coherent. This is formalized in \cref{prop:breg_cent}. 

\begin{proposition}\label{prop:breg_cent}
    Fix a model $M_k$, and let $p_{k,n}^\phi(\cdot \mid x_{1:n})$ be defined as in \cref{eq:breg_pred}. Suppose the Bregman divergences below are finite. Define
    \begin{equation*}
        \mathcal I_{k,n}^\phi = \int_{\Theta_k} D^\phi \big(f_k(\cdot; \theta_k) \, \big\| \, p_{k,n}^\phi(\cdot \mid x_{1:n}) \big) \pi_{k,n}^{\phi}(\theta_k \mid x_{1:n}) \, d\theta_k
    \end{equation*}
    Then, for any density $r$,
    \begin{equation*}
        \int_{\Theta_k} D^\phi \big(f_k(\cdot; \theta_k) \, \|\,  r \big) \pi_{k,n}^{\phi}(\theta_k \mid x_{1:n}) \, d\theta_k = \mathcal I_{k,n}^\phi + D^\phi \big(  p_{k,n}^\phi(\cdot \mid x_{1:n}) \, \| \, r \big).
    \end{equation*}
    Consequently, 
    \begin{equation*}
        p_{k,n}^\phi(\cdot \mid x_{1:n}) = \argmin_r \int_{\Theta_k} D^\phi \big(  f_k(\cdot; \theta_k) \, \| \, r \big) \pi_{k,n}^{\phi}(\theta_k \mid x_{1:n}) \, d\theta_k.
    \end{equation*}
    If $\phi$ is strictly convex, the minimizer is unique as a predictive density. The same identity holds for the leave-one-out posterior predictive $p_{k,n-1}^\phi(\cdot \mid x_{-i})$.
\end{proposition}

See \cref{supp:proof:breg_cent} for the proof.
\cref{prop:breg_cent} shows that the posterior predictive density $p_{k,n}^\phi(\cdot \mid x_{1:n})$ is not just a mixture obtained from the posterior. Under $D^\phi$, it is the Bayes action of the model densities $f_k(\cdot; \theta_k)$ when $\theta_k$ is distributed according to the Bregman posterior $\pi_{k,n}^\phi(\theta_k \mid x_{1:n})$. The quantity $\mathcal I_{k,n}^\phi$ measures the posterior average Bregman divergence from the model-specific predictive densities to their Bregman centroid. The term $D^\phi \big(  p_{k,n}^\phi(\cdot \mid x_{1:n}) \, \| \, r \big)$ is the excess posterior Bregman risk incurred by using an arbitrary density $r$ instead of the Bayes action. Thus, after updating by the Bregman score, the posterior predictive density is the decision-theoretically coherent predictive action under the corresponding Bregman divergence. 

\subsection{Generalized LOO predictive model selection (g-ELPD)}\label{sec:meth:gelpd}
Recall \cref{sec:pre:decision}, where predictive model comparison is formulated as maximizing expected utility. Given observed data $x_{1:n}$, we take this utility to be the same Bregman scoring rule $S^\phi$ that is used to define the Bregman posterior in \cref{sec:meth:breg_post}, so that parameter inference and predictive model evaluation are conducted by the same induced divergence. Define the expected predictive utility of model $M_k$ by
\begin{equation}\label{eq:exp_util}
    \bar u_{n,k} := \E\Big( S^{\phi}\bigl( p_{k,n}^{\phi}(\cdot \mid x_{1:n}),\,\tilde X \big) \Big),
\end{equation}
where $\tilde{X} \sim G$ is independent of $X_{1:n}$ and the expectation is taken over $\tilde X$. As reviewed in \cref{sec:pre:scoring}, since $S^\phi$ is strictly proper, maximizing $\bar u_{n,k}$ is equivalent to minimizing the induced divergence from $g$:
\begin{equation}
    \argmax_{k \in \{1,\ldots, K\}} \bar u_{n,k} = \argmin_{k \in \{1,\ldots, K\}} D^\phi\big( g \, \| \, p_{k,n}^{\phi}(\cdot \mid x_{1:n}) \big).
\end{equation}

Since $g$ is unknown, \cref{eq:exp_util} cannot be directly computed. Hence, we estimate it via leave-one-out cross-validation (LOO-CV). For each observation $i \in \{1,\ldots,n\}$, let $x_{-i} = (x_1,\ldots,x_{i-1},x_{i+1},\ldots,x_n)$ denote the leave-one-out subsample. Define the LOO posterior predictive density as:
\begin{equation}\label{eq:breg_pred_loo}
    p_{k,n-1}^\phi(x_i \mid x_{-i}) := \int_{\Theta_k} f_k(x_i; \theta_k)\; \pi_{k,n-1}^\phi(\theta_k \mid x_{-i}) \, d\theta_k,
\end{equation}
where $\pi_{k,n-1}^\phi(\theta_k \mid x_{-i})$ is the Bregman posterior computed from $x_{-i}$ via \cref{eq:breg_post}. The empirical LOO utility for model $M_k$ is then 
\begin{align}\label{eq:loo_utility}
    \hat{u}^{\mathrm{loo}}_{n,k} &:= \frac{1}{n} \sum_{i=1}^n S^\phi\big( p_{k,n-1}^\phi(\cdot \mid x_{-i}),  x_i \big).
\end{align}
We refer to \cref{eq:loo_utility} as the generalized ELPD (g-ELPD)\footnote{We define g-ELPD on a per-observation scale, i.e., as $\tfrac{1}{n}\sum_{i=1}^n$ of pointwise predictive scores; multiplying by $n$ yields the conventional summed ELPD used in \citet{vehtari2017practical}, without affecting model comparisons.} associated with the score $S^\phi$. When $\phi = u \log u$, the associated Bregman score is the log score (i.e., $S^\phi(p,x) = \log p(x)$), so \cref{eq:loo_utility} reduces to the usual LOO estimate of ELPD; for a general Bregman score, it estimates the corresponding expected out-of-sample predictive score. The resulting selection rule chooses the model with the highest estimated utility:
\begin{equation}\label{eq:select_rule}
    \widehat M_n := \argmax_{k \in \{1,\ldots, K\}} \hat{u}^{\mathrm{loo}}_{n,k}.
\end{equation}

Applying \cref{prop:breg_cent} to the leave-one-out posterior $\pi_{k,n-1}^\phi(\theta_k \mid x_{-i})$ shows that $p_{k,n-1}^\phi(\cdot \mid x_{-i})$ is the Bregman Bayes action for representing the model densities $f_k(\cdot; \theta_k)$ under $\pi^\phi_{k,n-1}(\theta_k\mid x_{-i})$. Therefore, each term in g-ELPD evaluates the Bregman-centroid predictive density for the corresponding held-out observation. 

This centering property is not a generic feature of arbitrary divergences. It is a characteristic feature of Bregman geometry: weighted averages are right Bregman centroids, and the associated Bregman information can be written as a divergence from the centroid \citep{banerjee2005clustering, Frigyik2008, Chodrow2025}. Thus, using a Bregman score allows us to retain the ordinary Bayesian posterior predictive mixture while giving it a decision-theoretic interpretation as the Bayes action under the same divergence used for updating and predictive evaluation. 

\subsection{Specialization to $\beta$-divergence}\label{sec:meth:beta}
A practically important and widely used choice is the $\beta$-divergence family in \cref{eq:beta_div}. When $\phi(p)=\tfrac{p^\beta}{\beta(\beta-1)}$, the associated Bregman scoring rule yields an explicit pointwise form of \cref{eq:loo_utility}: 
\begin{equation}\label{eq:loo_beta}
    \hat{u}^{\mathrm{loo}}_{n,k} = \frac{1}{n} \sum_{i=1}^n \bigg(\frac{p_{k,n-1}^\phi(x_i \mid x_{-i})^{\beta-1}}{\beta -1} - \frac{1}{\beta} \int_\cX p_{k,n-1}^\phi(t \mid x_{-i})^\beta \lambda(dt)  \bigg).
\end{equation}
The log-score case is recovered in the limit as $\beta \to 1$. We assume that the relevant $\beta$-power integrals are finite so that the score is well defined. The first term in \cref{eq:loo_beta} replaces the log predictive density by a power transformation, while the second integral term ensures the score is strictly proper. For $\beta>1$, observations assigned extremely small predictive density have reduced pointwise contribution relative to the log score. This weakens the influence of isolated outlying observations and tail mismatch in predictive comparison. \cref{prop:beta_robust} makes this robustness mechanism precise for pairwise score comparisons.

\begin{proposition}\label{prop:beta_robust}
    Let $\beta > 1$, and let $p$ and $q$ be densities with respect to $\lambda$. Suppose that $p$ and $q$ are uniformly bounded, i.e., $0 \le p(x), q(x) \le B$ for all $x \in \cX$, for some $B < \infty$, and that $\int_\cX p(t)^\beta \lambda(dt) < \infty$, $\int_\cX q(t)^\beta \lambda(dt) < \infty$. Define $\Delta^\beta_{p,q}(x) = S^\beta(p,x) - S^\beta(q, x)$. Then, for every $x \in \cX$,
    \begin{equation*}
        \lvert \Delta_{p,q}^\beta(x) \rvert \le C^\beta(p,q) ,
    \end{equation*}
    where 
    \begin{equation*}
        C^\beta(p,q) = \frac{2B^{\beta-1}}{\beta-1} + \frac{1}{\beta} \int_\cX \lvert p(t)^\beta - q(t)^\beta \rvert \lambda(dt).
    \end{equation*}
    In particular, the pointwise contribution of any single observation to the $\beta$-score comparison between $p$ and $q$ is uniformly bounded. By contrast, the log-score comparison $\log p(x) - \log q(x)$ is generally unbounded when either density can be arbitrarily close to zero. 
\end{proposition}

See \cref{supp:proof:beta_robust} for the proof. \cref{prop:beta_robust} gives a direct robustness mechanism for g-ELPD comparisons using $\beta$-divergence. A single low-density observation can dominate a log-score difference, whereas its contribution to a $\beta$-score comparison is bounded. Thus, g-ELPD using $\beta$-score reduces the extent to which pairwise model-comparison differences can be driven by isolated tail observations. A direct contamination bound follows from the same argument and is given as \cref{cor:contam_bound}. 

\begin{corollary}\label{cor:contam_bound}
    Let $g$ and $h$ be densities with respect to $\lambda$, and let $g_\epsilon = (1-\epsilon)g + \epsilon h$ where $0 \le \epsilon \le 1$. Under the conditions of \cref{prop:beta_robust},
    \begin{equation}
        \big\lvert \E_{g_\epsilon}\big(\Delta_{p,q}^\beta(X) \big) - \E_{g}\big(\Delta_{p,q}^\beta(X) \big) \big\rvert \le 2\epsilon C^\beta(p,q).
    \end{equation}
    Consequently, if $\E_g\big(\Delta_{p,q}^\beta(X)\big) > 2\epsilon C^\beta(p,q)$, then  $\E_{g_\epsilon}\big(\Delta_{p,q}^\beta(X)\big)>0$.
\end{corollary}
See \cref{supp:proof:contam_bound} for the proof. In words, the bound says that a small amount of contamination can only have a limited effect on a pairwise $\beta$-gELPD comparison. Thus, a positive pairwise $\beta$-score margin is preserved under $\epsilon$-contamination whenever the clean margin is larger than $2\epsilon C^\beta(p,q)$. In the LOO model comparison setting, this result applies conditionally to pairs such as $p = p_{k,n-1}^\beta (\cdot \mid x_{-i})$, $q = p_{l,n-1}^\beta(\cdot \mid x_{-i})$.  

\section{Asymptotic Theory}\label{sec:theory}
In this section, we study the large-sample behavior of generalized Bayesian updating based on the Bregman score $S^\phi$ and the resulting predictive model comparison procedure in the $\mathcal M$-open setting. General posterior concentration results for generalized Bayes posteriors are already well-established under regularity conditions (see, e.g., \citealt{miller2021asymptotic, martin2022direct}). Also, there is substantial work on robust Bayes parameter estimation using $\beta$-divergence, including asymptotics for estimators \citep{ghosh2016robust}. These results primarily concern within-model learning of parameters. Here, we connect such generalized updating to predictive model comparison across competing misspecified models. 

In \cref{sec:theory:model_selection}, we show that Bregman-score updating concentrates around the parameter value that minimizes the induced Bregman divergence to the data-generating process. We then show that the corresponding posterior predictive density converges to the pseudo-true predictive density. Finally, we prove that g-ELPD LOO evaluated using the same score consistently selects the model with the largest limiting Bregman predictive utility. 
In \cref{sec:theory:incoherent}, we study score-mismatched updating and evaluation and characterize the hybrid limiting target that results when the two scores differ.

Recall the notations in \cref{sec:pre}. We make the randomness in the sample explicit in order to state the limits and to distinguish random quantities from their deterministic targets. We first state the regularity conditions used in the asymptotic results. \cref{condition1} is a bounded-score regularity condition. It covers, for example, the $\beta$-score with $\beta>1$ when the model densities are bounded. The log score is recovered as a limiting case as $\beta \to 1$, but it is not covered by the bounded-score arguments below because $\phi'(u) = 1 + \log u$ is unbounded as $u \downarrow 0$.

\begin{condition}\label{condition1}
    Fix a model $M$ with density $f(\cdot; \theta)$ where $\theta \in \Theta \subset \bbR^{d}$. Assume $\Theta$ is compact. Let $X_1, X_2, \ldots$ be i.i.d. from the DGP $g$. Define
    \[
        h(\theta) = \E \big( S^{\phi}\big(f(\cdot; \theta), X\big)\big), \quad h_n(\theta) = \frac{1}{n}\sum_{i=1}^n S^{\phi}\big(f(\cdot; \theta), X_i \big),
    \]
    where $S^\phi$ is the Bregman scoring rule in \cref{def:bregman_score}. Assume:
    
    \item[(a)](unique pseudo-true parameter and separation) There exists a unique maximizer $\theta^* \in \mathrm{int}(\Theta)$ such that for every $\epsilon >0$,
    $$
        \sup_{\|\theta - \theta^*\| \ge \epsilon} h(\theta) < h(\theta^*).
    $$
    
    \item[(b)](Bounded model densities) There exists $C_f >0$ such that $0 \le f(x; \theta) \le C_f$ for all $x \in \cX$, $\theta \in \Theta$.
    
    \item[(c)](Regularity of $\phi$ for equicontinuity and finiteness) $\phi$ is strictly convex and continuously differentiable on $[0, C_f]$ with $\phi(0) =0$. Moreover, the integrand appearing in the normalization term of the Bregman score $m(u) := \phi'(u)u - \phi(u)$ is Lipschitz on $[0,C_f]$.
    
    
    \item[(d)](Uniform continuity in parameter) For any $\epsilon >0$, there exists $\delta >0$ such that if $\|\theta-\theta'\| < \delta$ then 
    \[
        \sup_{x \in \cX} |f(x; \theta) - f(x; \theta')| < \epsilon.
    \]
    
    \item[(e)](Prior positivity near $\theta^*$) The prior $\pi$ defined on $\Theta$ assigns strictly positive mass to every open neighborhood of the pseudo-true parameter $\theta^*$ defined in (a).
\end{condition}
We note that the compactness assumption is used only as a sufficient regularity condition for the concentration argument; the definition and computation of g-ELPD do not depend on this assumption.

Under \cref{condition1}, the \emph{pseudo-true parameter} for model $M_k$ is
\begin{equation}\label{eq:pseudo_true}
    \theta_k^* := \argmax_{\theta_k \in \Theta_k} \E \Big( S^\phi\big(f_k(\cdot; \theta_k), X\big)\Big) = \argmin_{\theta_k \in \Theta_k} D^\phi \big(g \| f_k(\cdot; \theta_k)\big).
\end{equation}
Define the limiting (population) predictive utility of model $M_k$ by 
\begin{equation}\label{eq:lim_ut}
    \bar{u}_k := \E \Big( S^\phi \big(f_k(\cdot; \theta_k^*), \tilde{X} \big) \Big).
\end{equation}
Thus, selecting the model with the largest $\bar{u}_k$ is equivalent to selecting the model whose best approximating density is closest to $g$ in Bregman divergence:
\begin{equation}\label{eq:ut_bd}
    \argmax_{k \in \{1,\ldots,K\}} \bar{u}_k = \argmin_{k \in \{1,\ldots,K\}} D^\phi\big(g \| f_k(\cdot; \theta_k^*)\big).
\end{equation}
We call any minimizer in \cref{eq:ut_bd} the optimal model under the scoring rule $S^\phi$.


\subsection{Model selection consistency}\label{sec:theory:model_selection}
We now state the asymptotic guarantees. For simplicity, we drop the model index $k$ in the statements below; the results hold for each $k$. First, \cref{thm:posterior_consistency} shows that the Bregman posterior puts asymptotically all its mass in any neighborhood of $\theta^*$. 

\begin{theorem}\label{thm:posterior_consistency}(Bregman Posterior Concentration)
    Assume \cref{condition1} (a)--(e). Then, for every $\eta >0$,
    \[
        \pi^\phi_{n} \big( \{\theta \in \Theta: \|\theta - \theta^* \| \ge \eta \} \mid X_{1:n} \big) \xrightarrow[n \to \infty]{\mathrm{a.s.}} 0.
    \]
\end{theorem}
All proofs are collected in the Supplementary Material \cref{supp:proofs}. Given that the posterior concentrates near $\theta^*$, the posterior predictive density becomes uniformly close to the model density at $\theta^*$. \cref{thm:ppc} connects the \cref{thm:posterior_consistency} to the predictive quantities used in model comparison. 

\begin{theorem}\label{thm:ppc}(Posterior predictive consistency)
    Assume Condition \ref{condition1} (a)--(e). Let $\pi_n^\phi(d\theta \mid X_{1:n})$ be the Bregman posterior defined in \cref{eq:breg_post}, and define the posterior predictive density by
    \[
        p_n^\phi(x \mid X_{1:n}) := \int_\Theta f(x; \theta) \pi_n^\phi(d\theta \mid X_{1:n}).
    \]
    Then,
    \[
        \sup_{x \in \cX} \big| p_n^\phi(x \mid X_{1:n}) - f(x; \theta^*) \big| \xrightarrow[n \to \infty]{\mathrm{a.s.}} 0.
    \]
\end{theorem}

Then, for each fixed model $M_k$ the leave-one-out g-ELPD estimator $\hat{u}^{\mathrm{loo}}_{n,k}$ in \cref{eq:loo_utility} is a consistent estimator of its population mean $\bar{u}_k$ (see \cref{supp:proof:loo_consistency}).  Finally, we state our main theorem of model selection consistency.

\begin{theorem}\label{sec:theory:consistency}(Model selection consistency)
    Assume that \cref{condition1} holds for each candidate model $M_k$, $k=1,\ldots, K$. Suppose that the limiting utilities $\{\bar u_k\}_{k=1}^K$ have a unique maximizer $M^* \in \cM$. Then, as $n \to \infty$,
    \[
         \widehat{M}_n := \argmax_{k \in \{1,\ldots,K\}} \hat{u}^{\mathrm{loo}}_{n,k} \, \xrightarrow[n \to \infty]{\mathrm{p}} \, \argmax_{k \in \{1,\ldots,K\}} \bar{u}_k := M^*. 
    \]
\end{theorem}
That is, the selected model $\widehat{M}_n$ converges in probability to the model whose best-fitting density minimizes the Bregman divergence to the true data-generating process $g$ as the number of data points increases. This establishes that our predictive model selection procedure is asymptotically consistent in targeting the optimal model under the Bregman scoring rule.


\subsection{Limiting target under score-mismatched updating and evaluation}\label{sec:theory:incoherent}

Our theoretical guarantees in \cref{sec:theory:model_selection} show that using the same Bregman scoring rule for parameter learning and predictive evaluation yields consistency for the optimal model under that score. However, it is natural to question if such a score-matched procedure is truly necessary. For instance, would combining Bregman-driven model selection with standard Bayesian parameter updating be sufficient? The following result shows that such a score-mismatched procedure is consistent, but it answers a different decision problem: the evaluation score is applied to the pseudo-true density induced by the updating score.

Consider two Bregman scores $S^{\phi_1}$ and $S^{\phi_2}$. For each model $M_k$, let $p^{\phi_2}_{k,n-1}(\cdot \mid X_{-i})$ denote the LOO posterior predictive computed from the $\phi_2$-Bregman posterior (i.e., the same construction as in \cref{eq:breg_pred_loo}, but with $\phi = \phi_2$). Define the score-mismatched LOO estimator
\begin{equation}\label{eq:uloo_incoherent_def}
    \hat{u}^{\mathrm{loo},(1|2)}_{n,k} := \frac{1}{n}\sum_{i=1}^n S^{\phi_1}\big(p^{\phi_2}_{k,n-1}(\cdot \mid X_{-i}), X_i\big).
\end{equation}
Let $\theta^{*(2)}_k$ denote the $\phi_2$-pseudo-true parameter for model $M_k$:
\begin{equation}
    \theta^{*(2)}_k := \argmax_{\theta_k \in \Theta_k} \E\Big(S^{\phi_2}\big( f_k(\cdot;\theta_k), X \big) \Big) = \argmin_{\theta_k \in \Theta_k} D^{\phi_2}\big( g \| f_k(\cdot; \theta_k) \big),
\end{equation}
in direct analogy with \cref{eq:pseudo_true}. Set
\begin{equation}\label{eq:uloo_incoherent_target}
    \bar{u}^{(1|2)}_{k} := \E\Big( S^{\phi_1}\big( f_k(\cdot; \theta^{*(2)}_k), X \big) \Big).
\end{equation}

\begin{theorem}\label{thm:loo_incoherent}
    For every $k$, assume that \cref{thm:ppc} holds for the updating score $S^{\phi_2}$. Also assume that the model densities are uniformly bounded as in \cref{condition1}(b), and that the evaluation score $S^{\phi_1}$ satisfies the regularity condition in \cref{condition1}(c). Then,
    \begin{itemize}
        \item[(a)] $\hat{u}^{\mathrm{loo},(1|2)}_{n,k} \xrightarrow[n \to \infty]{\mathrm{p}} \bar{u}^{(1|2)}_{k}$.
        \item[(b)] The limiting model comparison target is 
        \[
            \argmax_{k \in \{1,\ldots,K\}} \bar{u}^{(1|2)}_{k} = \argmin_{k \in \{1,\ldots,K\}} D^{\phi_1}\big(g \| f_k(\cdot; \theta^{*(2)}_k)\big).
        \]
    \end{itemize}
\end{theorem}
\cref{thm:loo_incoherent} implies that $\hat{u}^{\mathrm{loo},(1|2)}_{n,k}$ is a consistent estimator of the mismatched population utility $\bar{u}^{(1|2)}_{k}$. This utility evaluates the $\phi_1$-score at the $\phi_2$-pseudo-true density $f_k(\cdot; \theta_k^{*(2)})$, rather than at the $\phi_1$-pseudo-true density. Therefore, a score-mismatched procedure is consistent for its own limiting objective, but this objective need not select the model closest to $g$ under $D^{\phi_1}$, nor the model closest under $D^{\phi_2}$. This provides the formal reason for the score-matched construction used in the proposed g-ELPD criterion.

\section{Computation}\label{sec:comp}

A naive approach to evaluating \cref{eq:loo_utility} would require refitting the Bregman posterior $nK$ times, once for each held-out observation under each candidate model. Concretely, for each $i$ we would compute the LOO Bregman posterior $\pi_k^\phi(\theta_k\mid x_{-i})$ by Markov chain Monte Carlo or, in low-dimensional settings, deterministic grid evaluation. We would then form $p_{k,n-1}^\phi(\cdot\mid x_{-i})$, and evaluate $S^\phi\big(p_{k,n-1}^\phi(\cdot\mid x_{-i}),x_i\big)$. For many applications, this is computationally expensive for large $n$. Moreover, the LOO posteriors $\{\pi_k^\phi(\cdot\mid x_{-i})\}_{i=1}^n$ are often close to the full-data posterior $\pi_k^\phi(\cdot\mid x_{1:n})$, suggesting substantial  redundancy. The same redundancy motivates Pareto-smoothed importance sampling leave-one-out  (PSIS-LOO) cross-validation in standard Bayesian predictive model comparison \cite{vehtari2017practical, vehtari2024pareto}. We adapt this idea to the Bregman posterior in \cref{eq:breg_post}.

Our construction mirrors the standard PSIS-LOO, with the likelihood contribution used in ordinary Bayes replaced by the Bregman score contribution defining the generalized posterior. Under the ordinary Bayesian posterior with conditionally independent observations, the leave-one-out posterior satisfies $\pi(\theta \mid x_{-i}) \propto \pi(\theta \mid x_{1:n}) / f(x_i; \theta) $, which yields raw importance ratios proportional to $1/f(x_i; \theta)$. In the Bregman posterior, the analogous case-deletion identity is obtained directly from \cref{eq:breg_post}. For model $M_k$, the full-data Bregman posterior satisfies 
\begin{equation*}
    \pi_{k,n}^\phi(\theta_k \mid x_{1:n}) \propto \pi_k(\theta_k) \exp \Big( \sum_{m=1}^n S^\phi \big(f_k(\cdot; \theta_k), x_m) \Big).
\end{equation*}
Consequently, 
\begin{equation*}
    \pi_{k,n-1}^\phi (\theta_k \mid x_{-i}) \propto \pi_{k,n}^\phi (\theta_k \mid x_{1:n}) \exp \Big(\! -S^\phi \big(f_k(\cdot; \theta_k), x_i \big) \Big).
\end{equation*}
Thus, if $\theta_k^{(1)}, \ldots, \theta_k^{(J)}$ are draws from the full-data Bregman posterior, the raw case-deletion ratios are $\omega_{i,k}^{(j)} = \exp\Big(\! -S^\phi \big(f_k(\cdot; \theta_k^{(j)}), x_i \big) \Big)$, for $j=1,\ldots,J$. In the log-score special case, these reduce to $\omega_{i,k}^{(j)} \propto 1/f_k(x_i; \theta_k^{(j)})$. Although written for posterior draws, the same case-deletion identity also applies to deterministic grid approximations: grid weights replace Monte Carlo weights, and the ratios $\exp \Big(\! -S^\phi \big(f_k(\cdot; \theta_k), x_i \big) \Big)$, or their PSIS-smoothed versions, are applied before normalization.

Since these ratios can have heavy tails, PSIS replaces the upper tail of the empirical ratio distribution by a smoothed tail obtained from a fitted generalized Pareto distribution, producing stabilized importance weights and a diagnostic $\hat k$ that quantifies tail heaviness. For each $i$ and $k$, we apply Pareto smoothing to the raw ratios $\{\omega_{i,k}^{(j)}\}_{j=1}^J$ and normalize the smoothed ratios to obtain weights $\tilde \rho^{(j)}_{i,k}$, with $\sum_{j=1}^J \tilde \rho^{(j)}_{i,k} = 1$. The leave-one-out posterior predictive density is then approximated by 
\begin{equation*}
    \hat p^\phi_{k,n-1}(\cdot\mid x_{-i}) = \sum_{j=1}^J \tilde \rho^{j}_{i,k} f_k(\cdot;\theta_k^{(j)}).
\end{equation*}
The g-ELPD contribution is obtained by evaluating $S^\phi \big( \hat p^\phi_{k,n-1}(\cdot\mid x_{-i}), x_i \big) $. 

For the $\beta$-divergence family, the pointwise LOO contribution is therefore 
\begin{equation}\label{eq:beta_score_comp}
S^\beta\!\left(\hat p^\phi_{k,n-1}(\cdot \mid x_{-i}),x_i\right) 
= \frac{ \Big(\sum_{j=1}^J \tilde \rho^j_{i,k} f_k(x_i;\theta_k^{(j)})\Big)^{\beta-1}}{\beta-1} - \frac{1}{\beta} \int_\cX \Big( \sum_{j=1}^J \tilde \rho^j_{i,k} f_k(t;\theta_k^{(j)}) \Big)^\beta \lambda(dt).
\end{equation}
The first term requires density evaluations at the held-out observation $x_i$. The second term requires the $\beta$-power integral. If the observation space $\cX$ is discrete and finite, this integral is an exact finite sum over $\cX$. For countable $\cX$, we approximate the sum by truncating the support at negligible posterior predictive tail mass. For a continuous $\cX$, we approximate it numerically, for example by grid methods over the relevant support.

PSIS provides the tail-shape diagnostic $\hat k_{i,k}$ for each held-out observation and model. Following the standard PSIS-LOO practice, when $\hat k_{i} > \min(1 - \tfrac{1}{\log_{10}(J)}, 0.7)$ (typically for $J>2000$), the full-data Bregman posterior may be a poor proposal for the corresponding LOO posterior. When such cases occur, exact refitting for the problematic observations or $K$-fold cross-validation provides a more reliable alternative.

\section{Examples}\label{sec:examples}
\subsection{Contaminated Normal versus Miscentered Heavy-tailed Model}
In this simulation, we return to the motivating example from \cref{sec:intro} and study it as a full score-matched Bayesian predictive model comparison problem. Data are generated from the contaminated normal distribution $q(x) = (1-\epsilon)\cN(x; 0,1) + \epsilon\cN(x; 0, 10^2)$, with $\epsilon=0.2$. Thus, most observations come from a central $\cN(0,1)$ component, while a small fraction come from a much more dispersed contamination component. 

To demonstrate the utility of our g-ELPD framework, we compare two misspecified candidate models. The first model is a normal location model: $X_i \mid \mu \sim \cN(\mu,1)$, with prior distribution $\mu \sim \cN(0,5^2)$. This model can learn the central location of the bulk of the data, but it cannot represent the heavy contamination tail. The second model is a miscentered heavy-tailed model: $X_i \mid \nu \sim t_\nu(3,1)$, with prior distribution $\nu -1 \sim \mathrm{Lognormal}(\log 5, 1)$. This model can adapt its tail thickness through $\nu$, but its location is fixed at $3$ and is therefore systematically shifted away from the central component of the data. We consider several values of $\beta$ close to one to examine how small departures from the log score change the comparison. 


Computation of the g-ELPD for this example proceeds as follows. For each value of $\beta$, both models are updated using the corresponding $\beta$-score Bregman posterior and evaluated using the same $\beta$-score LOO predictive utility. Since both models have one-dimensional parameter spaces, we use deterministic grid approximations. For the normal, $\beta=1$ uses the conjugate closed-form LOO posterior predictive distribution, while $\beta>1$ uses direct LOO calculation on the grid $\mu \in [-6, 6]$ with $501$ equally spaced points. 
For the $t$ model, we work on the grid $\eta = \log(\nu-1) \in [\log(0.01), \log(79)]$, also with $501$ equally spaced points. The t-model LOO contributions are first computed by PSIS using the generalized case-deletion ratios induced by the $\beta$-posterior; since the posterior is represented on a grid, the ratios are applied to the full-data grid posterior weights. When the Pareto-$\hat k$ diagnostic exceeds $1-1/\log_{10}(501)$, the corresponding pointwise LOO contribution is replaced by a direct deterministic grid calculation over $\eta$. For $\beta>1$, the integral term in \cref{eq:loo_beta} is approximated by a rectangular grid over $z \in [-120,120]$ with $1001$ equally spaced points. 

\begin{figure}
    \centering
    \includegraphics[width=\linewidth]{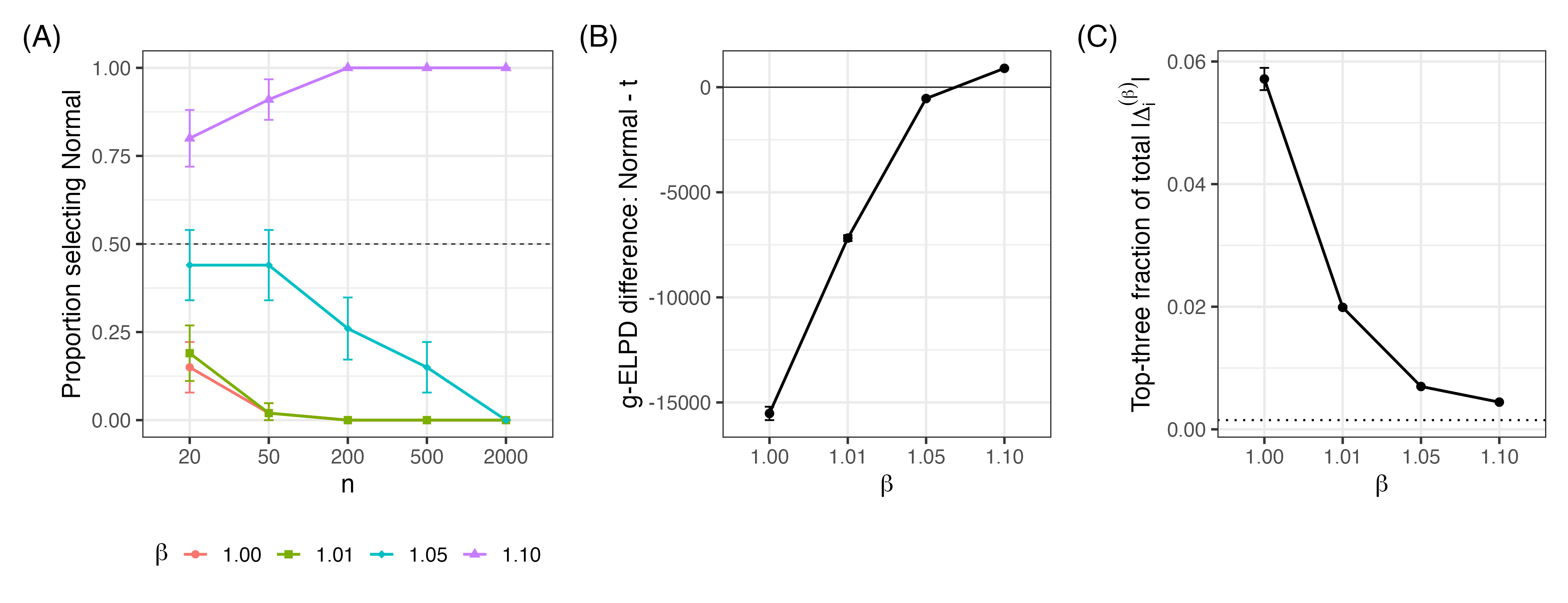}
    \captionsetup{font=small}
    \caption{Robust predictive model selection in the contaminated normal simulation. (A) Proportion of 100 replications in which the normal model has larger g-ELPD, across sample sizes and values of $\beta$. Error bars show two Monte Carlo standard errors. (B) Mean signed pairwise g-ELPD difference, normal minus t, for $n=2000$. Positive values favor the normal model. (C) Mean fraction of the total absolute pointwise difference contributed by the three observations with the largest $\lvert \Delta_i^{(\beta)} \rvert$, for $n=2000$. The dotted horizontal line marks the lower bound $3/n$. Increasing $\beta$ reduces the concentration of the comparison among a few
    extreme observations and shifts selection toward the normal model.}
    \label{fig:normal_t_sim}
\end{figure}

We consider sample sizes $n \in \{20, 50, 200, 500, 2000\}$, and $\beta \in \{1, 1.01, 1.05, 1.1\}$. For each pair $(n,\beta)$, we run $100$ independent replications. In each replication, we compute the g-ELPD for both candidate models and record whether the normal model has the larger value. \cref{fig:normal_t_sim} summarizes these selection frequencies and the corresponding pairwise comparison mechanism.
Panel (A) shows how the selected model changes with both $n$ and $\beta$. As $n$ grows, the log-score comparison increasingly selects the t model, whereas the larger values of $\beta$ increasingly select the normal model. Thus, the selected model depends on the scoring rule used in the g-ELPD criterion.
Panel (B) focuses on the signed g-ELPD difference, normal minus t, for $n=2000$. This shows the margin of the pairwise predictive utility. As $\beta$ increases, the signed contrast moves from strongly favoring the t model toward favoring the normal model, showing that the change in selection is driven by a change in the pairwise predictive utility.

Let $\Delta_i^{(\beta)} = S^\phi (p^\beta_{\mathrm{N},-i}, x_i) - S^\phi (p^\beta_{\mathrm{T},-i}, x_i)$, where $p^\beta_{\mathrm{N},-i}$ and $p^\beta_{\mathrm{T},-i}$ are the leave-one-out score-matched posterior predictive densities under the normal and t models, respectively. Panel (C) shows the fraction of the total absolute pairwise contrast contributed by the three observations with the largest $\lvert \Delta_i^{(\beta)} \rvert$, again for $n=2000$. This fraction decreases substantially as $\beta$ increases, showing that the change in the pairwise g-ELPD margin is accompanied by a reduction in the concentration of the comparison among a few extreme observations. Thus, the example demonstrates the bounded-contribution mechanism in \cref{prop:beta_robust}: increasing $\beta$ reduces the influence of low-predictive-density observations and shifts the score-matched comparison away from tail protection alone.

\subsection{Thermal Performance Curves in Microbial Ecology}\label{sec:realdata}
Modeling bacterial growth rate as a function of temperature is important for understanding how microbial populations respond to environmental change. 
Thermal Performance Curve (TPC) models are widely used for this study, which describe how biological performance increases with temperature up to an optimum and then declines near the upper thermal limit \citep{sinclair2016can}. 
A key practical issue is that different TPC models can give similar fits over the central temperature range but differ substantially in how they represent behavior near the lower and upper temperature limits. This makes TPC model selection a natural setting for evaluating robust predictive criteria.

We analyze the \textit{Pseudomonas putida} specific growth-rate dataset from \citet{ratkowsky2005unifying}, included in the thermal performance compilation of \citet{kontopoulos2024no}. The dataset contains $n=85$ observations across temperatures ranging from $-0.39^\circ$C to $39.21^\circ$C. Let $T_i$ denote the temperature and $y_i$ the corresponding scaled growth-rate observation. Many studies show that no single TPC model is universally optimal across traits and taxa despite the proliferation of their use \citep{kontopoulos2024no, kellermann2019comparing, angilletta2006estimating}, motivating dataset-specific model comparison among plausible TPC families.

We compare three representative low-dimensional TPC models from \citet{kontopoulos2024no}: Bilinear, Simplified Briere I, and Analytis--Kontodimas. The Bilinear model can adapt sharply to behavior near the bounds of the temperature range, whereas the other two models are smooth, bounded TPC models with lower and upper thermal limits. The three mean functions are given in \cref{tab:tpc-models}. This contrast is biologically meaningful for bacterial-specific growth rates, which may exhibit gradual changes over the central temperature range but rapid decline near thermal limits (\cref{fig:tpc_fits}(A)). 

\begin{table}
\centering
\captionsetup{font=small}
\renewcommand{\arraystretch}{1.35}
\setlength{\tabcolsep}{8pt}
\begin{tabular}{lcl}
\toprule
Model & Parameters & Mean function $B(T)$ \\
\midrule
Bilinear & 4 & $ B(T) =
\begin{cases}
B_{pk} \cdot \dfrac{T-T_{\min}}{T_{\mathrm{pk}}-T_{\min}}, & T_{\min}<T\le T_{\mathrm{pk}},\\
B_{pk} \cdot \dfrac{T_{\max}-T}{T_{\max}-T_{\mathrm{pk}}}, & T_{\mathrm{pk}}<T<T_{\max}.
\end{cases}
$ \\

Simplified Briere I & 3 & $a (T-T_{\min})\sqrt{T_{\max}-T}$ \\

Analytis--Kontodimas & 3 & $a(T-T_{\min})^2(T_{\max}-T)$ \\
\bottomrule
\end{tabular}

\vspace{0.5em}
\caption{Candidate thermal performance curve models used in the \textit{Pseudomonas putida} application. Here $B(T)$ denotes the mean scaled growth rate at temperature $T$; $T_{\min}$, $T_{\max}$, $T_{\mathrm{pk}}$ denote the lower thermal limit, upper thermal limit, and peak temperature, respectively; $B_{\mathrm{pk}}$ denotes the peak growth rate; and $a >0$ is a scale parameter. For all three models, the displayed function is used for $T_{\min} < T < T_{\max}$, with $B(T)=0$ outside this interval.}
\label{tab:tpc-models}
\end{table}

We fit each model in Stan \citep{carpenter2017stan} using the score-matched Bregman posterior and compute g-ELPD LOO by the PSIS approximation described in \cref{sec:comp}. All Stan fits showed good convergence diagnostics: all $\hat R$ values were below $1.01$ and effective sample sizes were large enough across models and values of $\beta$. PSIS diagnostics $\hat k$ were generally stable across $\beta$.

\begin{figure}
    \centering
    \includegraphics[width=\linewidth]{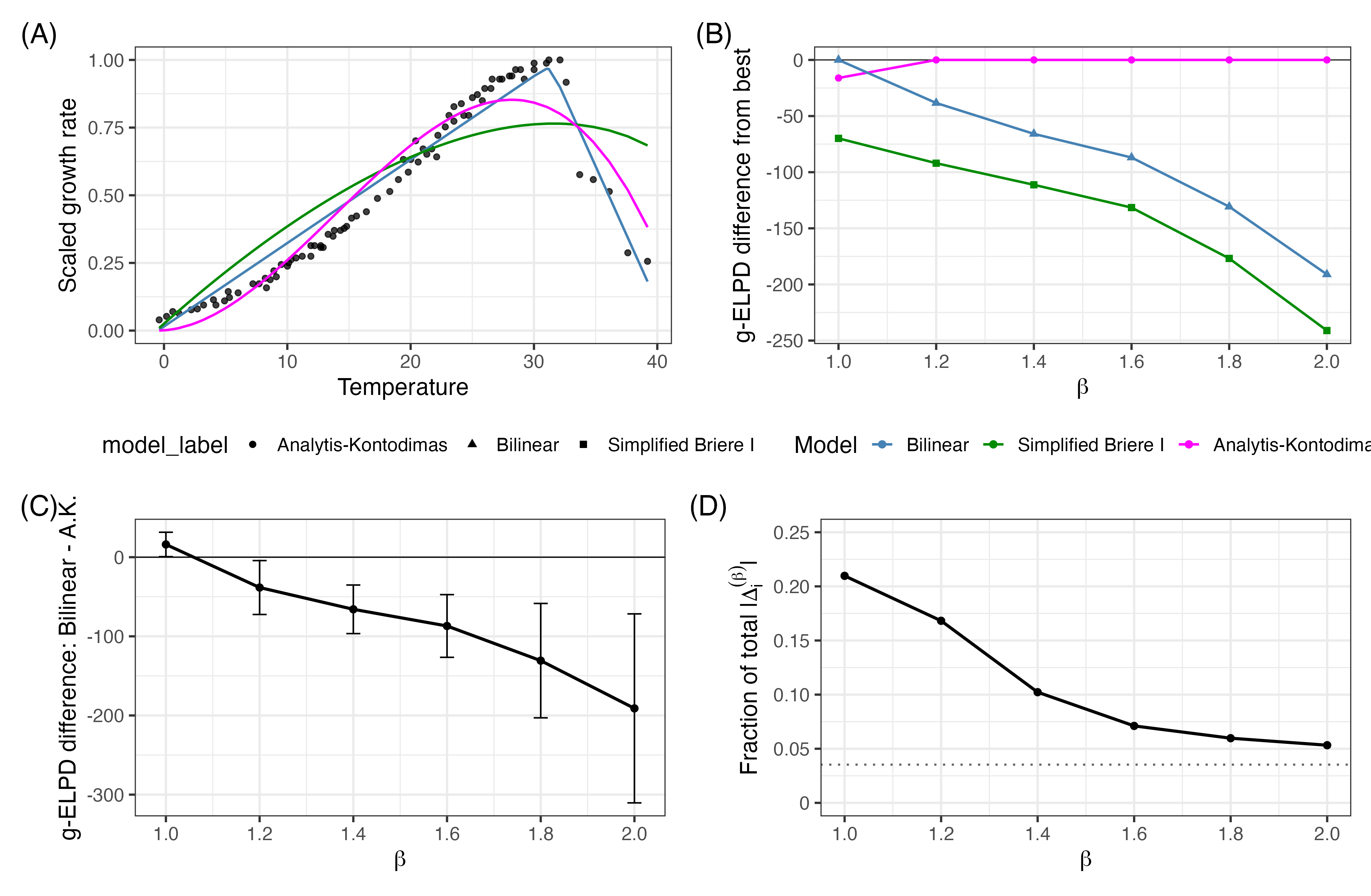}
    \captionsetup{font=small}
    \caption{Robust predictive model selection for the \textit{Pseudomonas putida} TPC dataset. (A) Normalized growth-rate observations and posterior mean fits of the three candidate TPC models under the log score ($\beta=1$); (B) Model-specific g-ELPD differences from the best model at each $\beta$; (C) Signed pairwise g-ELPD difference of Bilinear minus Analytis--Kontodimas, with $2$ times the standard error; positive values favor Bilinear and negative values favor Analytis--Kontodimas; (D) Fraction of the total absolute pointwise difference contributed by the three observations with the largest $|\Delta_i^{(\beta)}|$. The horizontal dotted line marks the lower bound, $3/85$, corresponding to equal absolute contributions across all observations. Standard ELPD selects Bilinear at $\beta=1$, whereas g-ELPD selects Analytis--Kontodimas for all $\beta > 1$ considered. The reduction in the largest pointwise contributions is consistent with \cref{prop:beta_robust}}
    \label{fig:tpc_fits}
\end{figure}

\cref{fig:tpc_fits} summarizes the TPC comparison\footnote{In this section, we report g-ELPD differences on the conventional summed scale}. Panel (B) shows model-specific g-ELPD differences from the best model at each value of $\beta$. Under the log score, Bilinear has the largest predictive utility, while Analytis--Kontodimas is the second-best. For every $\beta > 1$ considered, however, Analytis--Kontodimas becomes the selected model with the g-ELPD gap increasing as $\beta$ grows. Simplified Briere I is included as a standard smooth TPC alternative, but it is not competitive for this dataset.

To understand why the selected model changes, we examine the pairwise comparison between Bilinear and Analytis--Kontodimas. Let $\Delta_i^{(\beta)} = S^\phi (p^\beta_{\mathrm{Bilinear},-i}, y_i) - S^\phi (p^\beta_{\mathrm{AK},-i}, y_i)$, where $p^\beta_{\mathrm{Bilinear},-i}$ and $p^\beta_{\mathrm{AK},-i}$ are the leave-one-out Bregman posterior predictive densities under the Bilinear and Analytis--Kontodimas models. \cref{fig:tpc_fits}(C) shows the signed pairwise g-ELPD difference, Bilinear minus Analytis--Kontodimas, with two standard-error bars. Positive values favor Bilinear and negative values favor Analytis--Kontodimas. At $\beta=1$, the difference is $16.09$ with standard error $7.70$. Once $\beta>1$, the sign reverses and the estimated comparison consistently favors Analytis--Kontodimas.

\cref{fig:tpc_fits}(D) shows the fraction of the total absolute pairwise contrast contributed by the three observations with the largest $|\Delta_i^{(\beta)}|$.  Under the log score, these three observations account for $21\%$ of the model comparison. This fraction decreases steadily as $\beta$ increases, reaching $5.3\%$ at $\beta=2$. Thus, the change in model ranking is accompanied by a reduction in the concentration of the pairwise comparison among a few high-leverage observations. This is the empirical mechanism predicted by \cref{prop:beta_robust}.

\subsection{Randomly Acquired Characteristics on Shoe Soles}

An important statistical problem in the evaluation of forensic footwear evidence is modeling the spatial distribution of randomly acquired characteristics (RACs) -- the cuts, holes, and other forms of damage that accumulate on shoe soles through use \citep{KELLETT2026100673}. The locations of RACs play a key role in distinguishing prints made by shoes of the same make, brand, and size. Inaccurate models may underestimate the probability of coincidental matches and thereby overstate the probative value of footwear evidence.

\begin{figure}[H]
    \centering
    \includegraphics[width=0.7\linewidth]{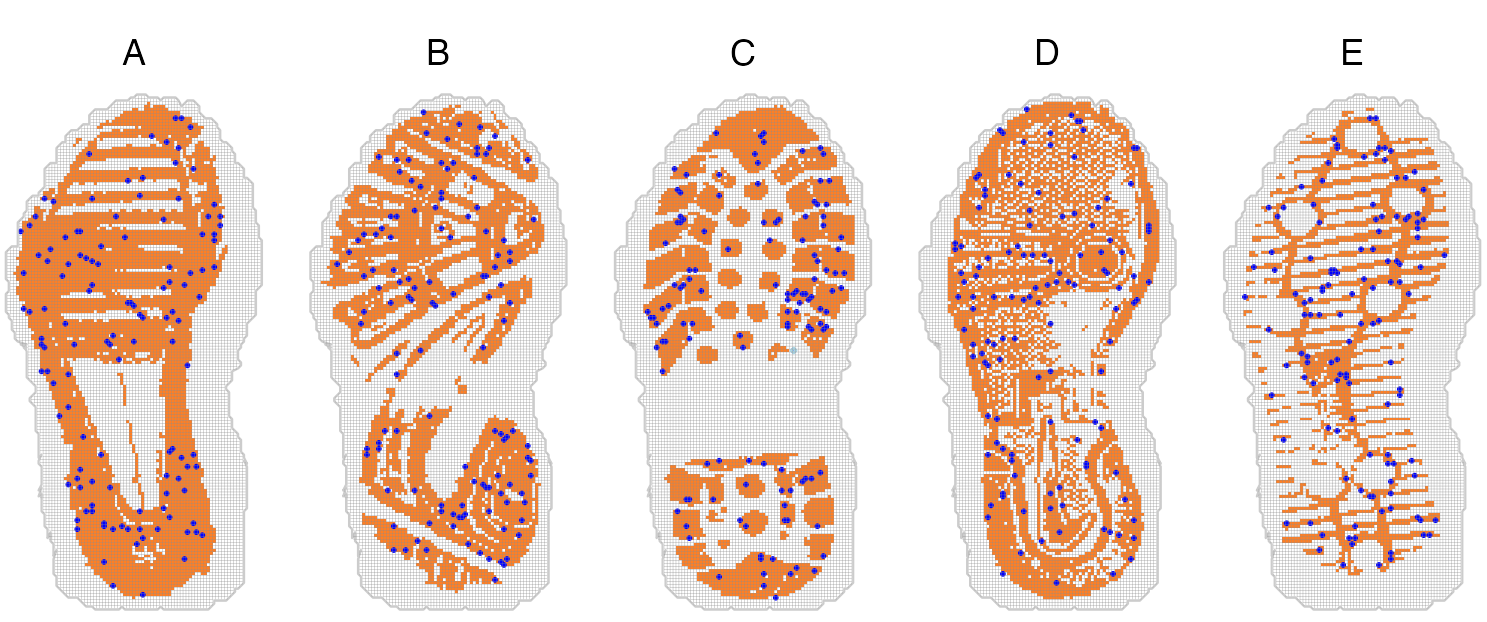}
    \captionsetup{font=small}
    \caption{The binarized contact grid and RAC locations for the five shoe treads. Orange tiles depict the presence of contact surface. Blue points indicate the RAC locations, with light blue distinguishing a single RAC located nonadjacent to contact surface.}
    \label{fig:shoe_data}
\end{figure}

In this study, we consider a synthetic version of the JESA database \citep{wiesner2020dataset}, made publicly available by \cite{spencer2020bayesian}, to perform a g-ELPD comparison of competing models for the spatial distribution of RACs. Figure~\ref{fig:shoe_data} presents the five shoe treads $s \in \left\{A, B, C, D, E\right\}$ for which data are available. Each tread is represented as a binarized image of a laboratory-generated test impression of the shoe, standardized to a common grid consisting of $J = 11472$ locations. This representation is treated as the shoe's contact surface: a value of 1 indicates that the shoe sole contacts the ground (shown in orange), and 0 indicates that it does not (shown as white). For each tread $s$, 125 RAC coordinates $y_{1,s}, \ldots, y_{125,s}$ are available (shown in blue) for a total of 625 observed locations.

We consider the task of modeling the shoe-specific spatial distributions $\lambda_s$ of RAC locations for each shoe $s$, where $y_{i,s} \overset{\text{i.i.d.}}{\sim} \lambda_s$. By convention, $\lambda_s$ is assumed to be piecewise constant at the resolution of the contact surface. Inferring $\lambda_s$ thus amounts to estimating a discrete distribution over the $11472$ locations. Four candidate models for this task are outlined in Table~\ref{tab:shoe-models}, formulated as generalized linear models following \citet{manna2026scalable}.

\begin{table}[H]
\centering
\captionsetup{font=small}
\renewcommand{\arraystretch}{1.35}
\setlength{\tabcolsep}{8pt}

\begin{tabular}{l c l @{\;}l}
\toprule
Model & Parameters & \multicolumn{2}{l}{RAC Probabilities} \\
\midrule

Uniform & 0
& $\lambda_s(y)$
& $\propto 1$
\\

Spatial & 138
& $\lambda_s(y)$
& $\propto \exp\left(\alpha^{\text{spatial}}_{r(y)}\right)$
\\

Contact & 6
& $\lambda_s(y)$
& $\propto \exp\left(\alpha^{\text{contact}}_{1 + c_s(y)}\right)$
\\
Spatial + Contact & 144
& $\lambda_s(y)$
& $\propto \exp\left(
\alpha^{\text{contact}}_{1 +c_s(y)}
+
\alpha^{\text{spatial}}_{r(y)}
\right)$
\\
\bottomrule
\end{tabular}
\vspace{0.5em}
\caption{Candidate Models for the RAC distributions. 
}
\label{tab:shoe-models}
\end{table}
Here, $\alpha^{\text{spatial}} \in \mathbb{R}^{138}$ captures a spatial trend in RAC probabilities that is shared across shoes. The function $r$ maps each location $y$ to one of 138 distinct $10 \times 10$ regions within which the spatial effect is assumed constant. An ICAR \citep{besag1991bayesian} prior is placed on $\alpha^{\text{spatial}}$ to promote spatial smoothness, with neighboring regions defined as adjacent. Similarly, $\alpha^{\text{contact}} \in \mathbb{R}^{6}$ models the effect of the contact surface on RAC probabilities. The function $c_s(y) \in \left\{0,\ldots,5\right\}$ counts the number of contact surface locations among $y$ and its four immediate neighbors. To encode an assumption that RACs can be observed only on the contact surface \citep{kaplan2022location}, $\alpha^{\text{contact}}_1$ is fixed at $-500$, effectively assigning zero probability wherever $c_s(y) = 0$. The remaining five effects are assigned independent $\text{Normal}(0, 6^2)$ priors.

We fit each model in Stan \citep{carpenter2017stan} using the score-matched Bregman posterior, achieving satisfactory performance on Stan's MCMC diagnostics. We then computed g-ELPD LOO using the PSIS approach described in \cref{sec:comp}. Because $y$ takes values in a finite grid, the sum for the second term in \cref{eq:beta_score_comp} can be computed without truncation.

The PSIS diagnostic $\hat k$ exceeded the threshold for a single RAC on shoe C under the Spatial + Contact model. To ensure an accurate LOO score for this observation, we reran the Stan code with this RAC held-out and directly computed its g-ELPD LOO score from the resulting posterior draws. The offending observation was the only RAC in the dataset located at a location for which $c_s(y) = 0$. It can thus be viewed as an outlier due to its violation of the assumption that RACs only appear on the contact surface. In Figure~\ref{fig:shoe_data}, this RAC can be seen near the lower-right edge of the forefoot on Shoe C, distinguished using light blue.

\begin{figure}
\centering

\begin{subfigure}{0.48\linewidth}
    \caption{Predictive probabilities under $\beta = 1$}
    \label{sub:shoe_predictives2}
    \includegraphics[width=\linewidth]{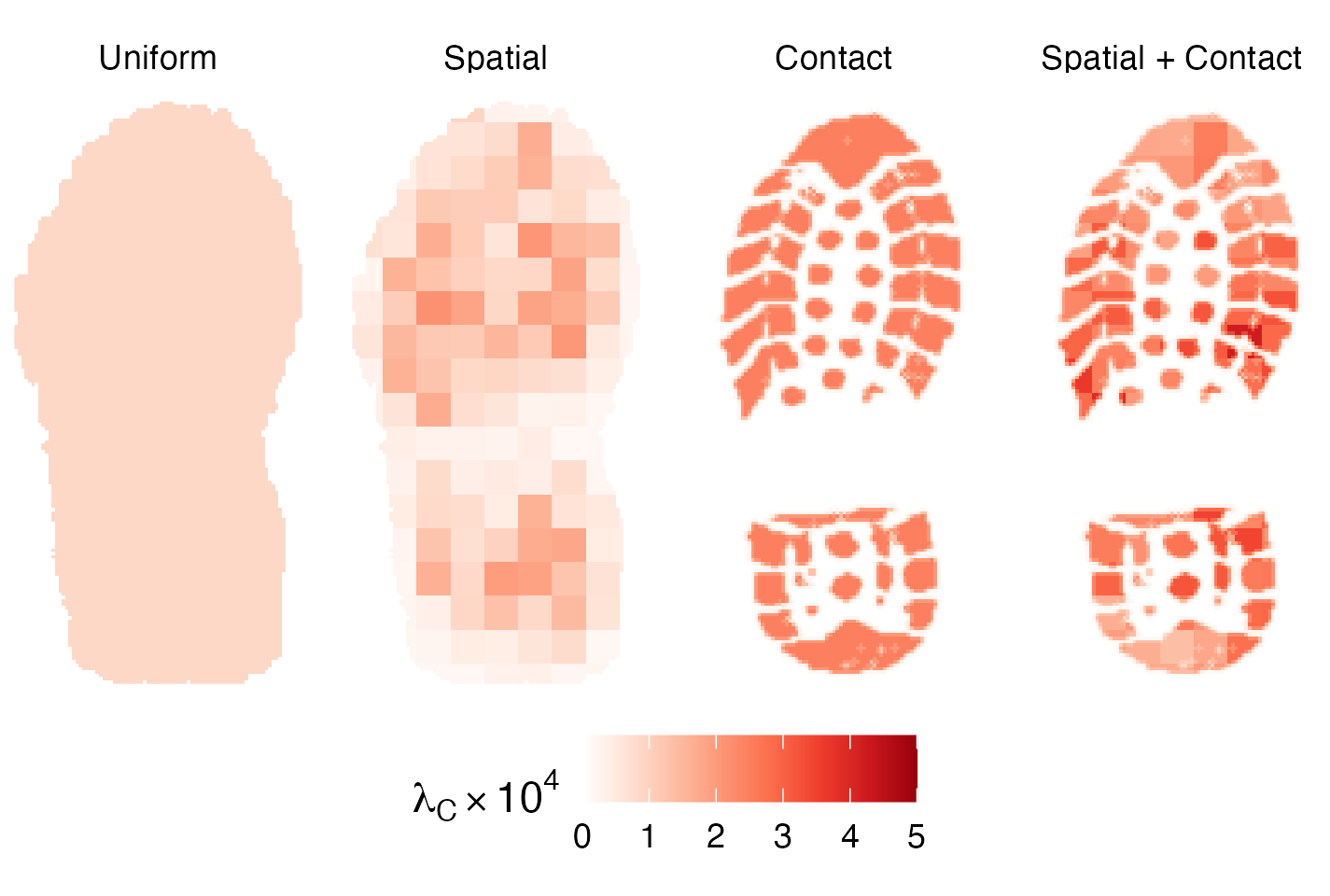}
\end{subfigure}\hfill
\begin{subfigure}{0.48\linewidth}
    \caption{Model-specific g-ELPD differences}
    \label{sub:shoe_comparison}
    \includegraphics[width=\linewidth]{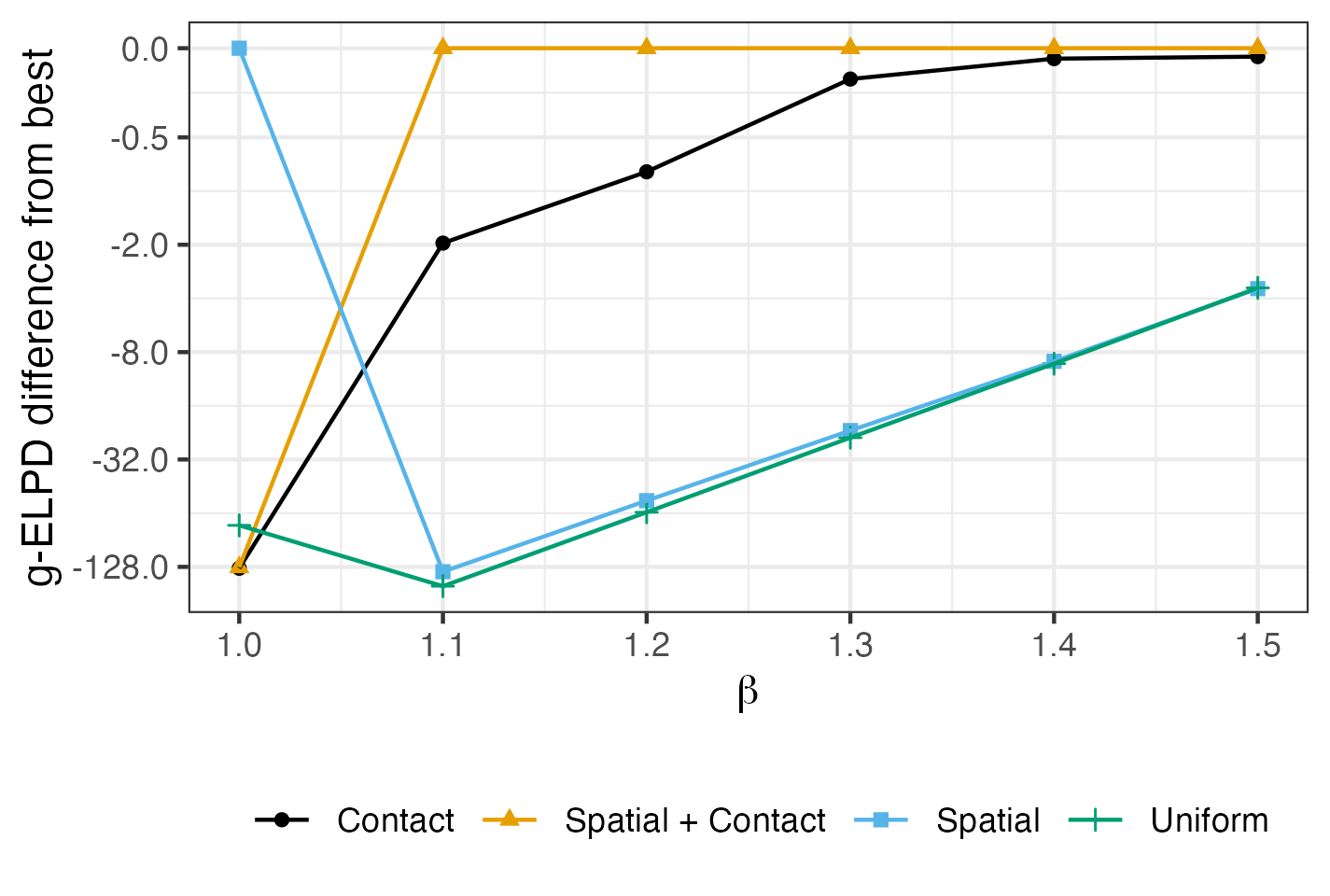}
\end{subfigure}

\vspace{1em}

\begin{subfigure}{0.48\linewidth}
    \caption{Spatial versus Spatial + Contact}
    \label{sub:shoe_minus}
    \includegraphics[width=\linewidth]{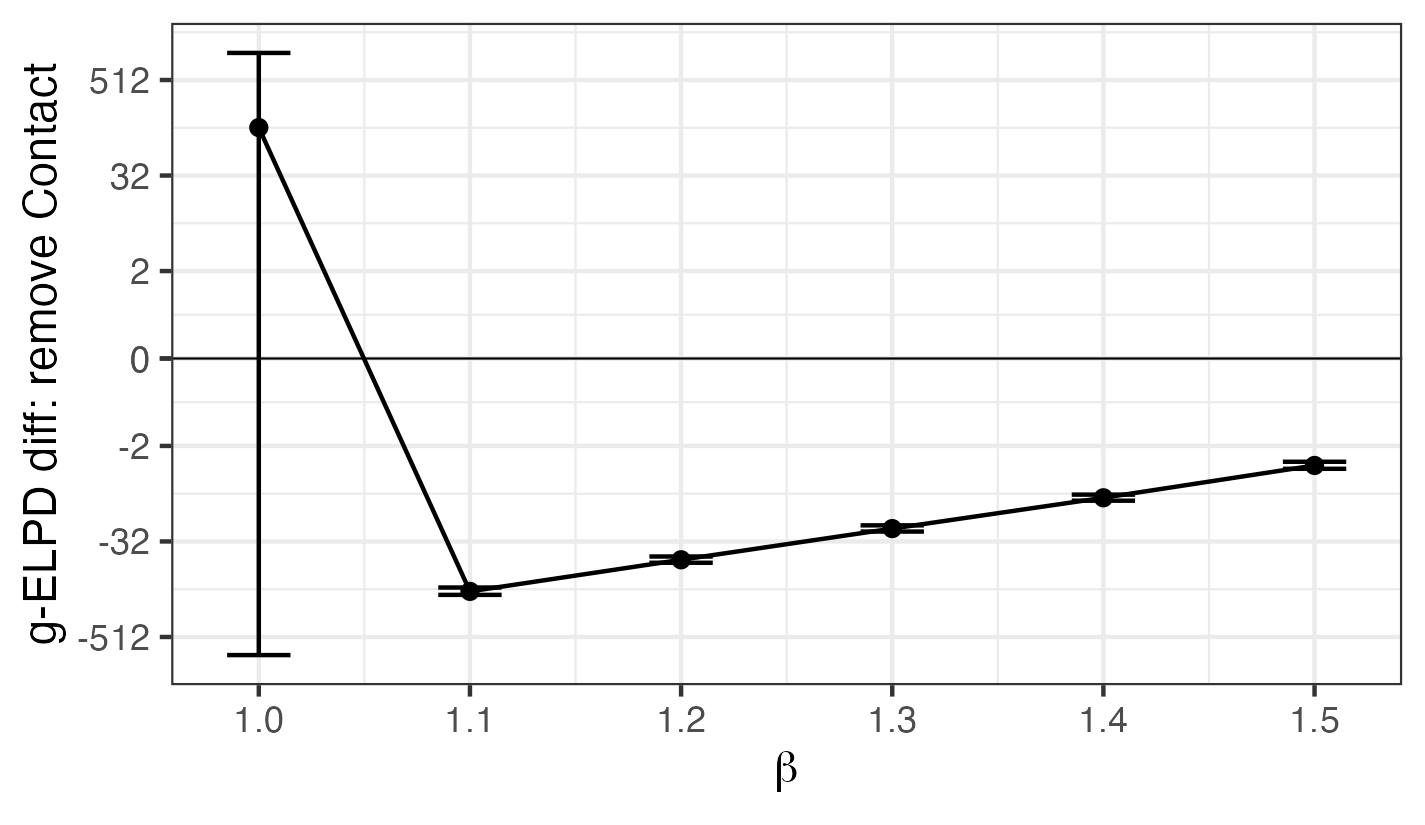}
\end{subfigure}\hfill
\begin{subfigure}{0.48\linewidth}
    \caption{Contribution of the largest $|\Delta_i^{(\beta)}|$}
    \label{sub:shoe_fraction}
    \includegraphics[width=\linewidth]{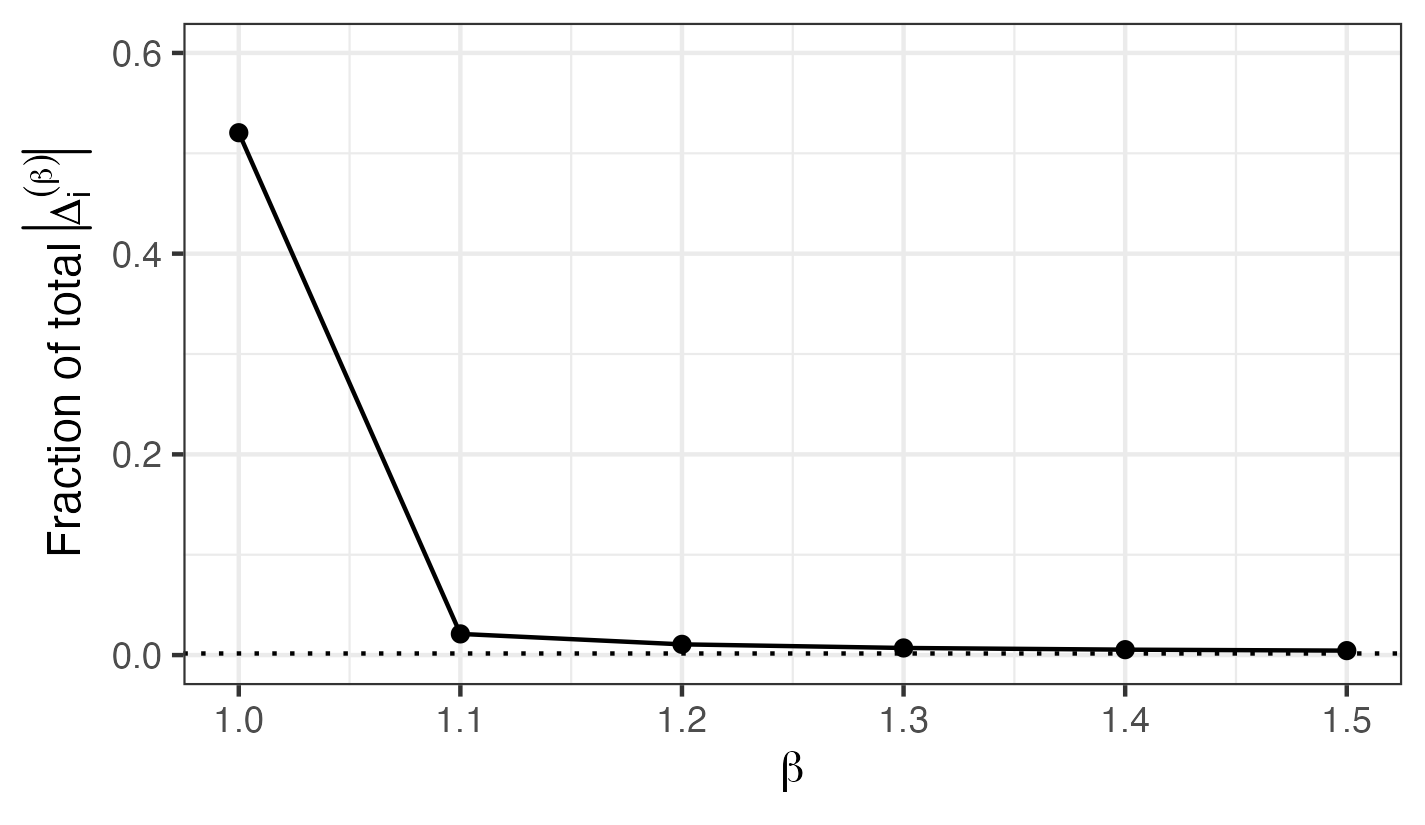}
\end{subfigure}
\captionsetup{font=small}
\caption{Robust predictive model selection for the synthetic JESA dataset. (a) Posterior mean fits of $\lambda_{\text{C}}$ for the four candidate models under the log score ($\beta = 1$). (b) Model-specific g-ELPD differences from the best model at each $\beta$, shown on a logarithmic scale; (c) Signed pairwise g-ELPD difference of the Spatial + Contact model versus the Spatial Model, with $2$ times the standard error; (d) Fraction of the total absolute pointwise difference contributed by the observation with the largest $|\Delta_i^{(\beta)}|$; the horizontal dotted line marks the lower bound, $1/625$, corresponding to equal contributions across all observations. Standard ELPD selects the Spatial Model at $\beta=1$, whereas g-ELPD selects Spatial + Contact for all $\beta > 1$ considered. The reduction in the largest pointwise contribution is consistent with \cref{prop:beta_robust}.
}
\label{fig:shoe_main_grid}
\end{figure}


\cref{fig:shoe_main_grid} summarizes the results of the g-ELPD model comparison. Under the traditional log score ($\beta = 1$), models without the contact surface component perform better, whereas models incorporating the contact surface dominate for all $\beta > 1$ (\cref{sub:shoe_comparison}). \cref{sub:shoe_minus} reinforces this finding for the two models incorporating the spatial component, and highlights the high degree of uncertainty in the model comparison for $\beta = 1$. \cref{sub:shoe_fraction} then establishes that the majority of this absolute performance difference at $\beta = 1$ is attributed to a single RAC, thereby explaining the high degree of uncertainty.

Indeed, the RAC responsible for these differences is the same outlier observation on shoe C discussed above. The two contact surface models assign effectively no probability to locations for which $c_s(y) = 0$ (\cref{sub:shoe_predictives2}), whereas the Uniform and Spatial models do not use the contact surface. Consequently, under the log score, the contact models perform worse (albeit with a great deal of uncertainty) due to the very large penalty incurred for poor performance on this single RAC. In contrast, for the g-ELPD scores with $\beta > 1$, the impact of any single observation is bounded, limiting the penalty for the outlier. This reveals that the contact surface models exhibit superior performance for the vast majority of RACs. The change in rankings as $\beta$ increases demonstrates \cref{prop:beta_robust}.

A secondary finding of this analysis is that models which incorporate the spatial component outperform those that do not, but the relative advantage is attenuated for larger values of $\beta$. This suggests that
the implicit penalty on location-wise probability variation implied by the second term in \cref{eq:beta_score_comp} outweighs the performance gains obtained by spatial variation as $\beta$ increases, flattening the optimal spatial effect. This effect is demonstrated directly by \cref{fig:shoe_predictives3} in the Supplementary Material.


\section{Discussion}\label{sec:conclusion}
We developed a score-matched Bayesian predictive model selection framework in which the same Bregman scoring rule is used for generalized posterior updating, posterior predictive construction, and LOO predictive evaluation through g-ELPD. This framework recovers standard ELPD under the log score, while allowing model selection to target different predictive pseudo-truths under different Bregman divergences; for the $\beta$-divergence family, $\beta>1$ reduces the influence of isolated low-density observations. In practice, we view the choice of $\beta$ as part of the broader choice of $\beta$-score, rather than as a universal tuning parameter to be optimized. Therefore, we recommend reporting g-ELPD over a range of $\beta$ values, including $\beta=1$ as the log-score baseline, and interpreting changes in model ranking as sensitivity to the chosen divergence.

We demonstrate the practical efficacy of our approach on three separate examples. However, some settings may require additional methodological advances for practical implementation.
First, for $\beta>1$, g-ELPD requires evaluating the $\beta$-score integral, which may be difficult in multivariate, discrete-continuous, or high-dimensional observation spaces. Second, although we report conventional LOO standard errors, uncertainty quantification for generalized-score model comparisons remains less developed than for standard ELPD and should be interpreted cautiously in finite samples.

Future work should develop scalable approximations for the Bregman score integral, including Monte Carlo, quadrature, or variational approximations that remain stable in higher dimensions. Another important direction is uncertainty quantification for generalized-score predictive comparisons, for example through Bayesian-bootstrap or finite-sample approaches for pointwise score differences. A further open problem is principled divergence selection. The $\beta$-divergence family provides a robustness path, but other Bregman divergences may be better aligned with different scientific goals.  Finally, the same score-matched idea could be extended from winner-take-all selection to predictive combination, such as Bregman stacking or Bayesian-bootstrap model weights under general proper scoring rules.

\section*{Acknowledgments}
We thank Zhiyi Chi, Gyuhyeong Goh, Steve MacEachern, Jacob Fontana, Jack Jewson, David Frazier, Colin Kremer, and Jeff Miller for helpful conversations.
N.A.S. was supported in part by the Makuch Faculty Fund Award in Mathematical and Data Sciences at University of Connecticut.

\putbib
\end{bibunit}



\clearpage
\setcounter{page}{1}
\setcounter{section}{0}
\setcounter{table}{0}
\setcounter{figure}{0}
\setcounter{algorithm}{0}
\renewcommand{\theHsection}{SIsection.\arabic{section}}
\renewcommand{\theHtable}{SItable.\arabic{table}}
\renewcommand{\theHfigure}{SIfigure.\arabic{figure}}
\renewcommand{\theHalgorithm}{SIalgorithm.\arabic{algorithm}}
\renewcommand{\thepage}{S\arabic{page}}  
\renewcommand{\thesection}{S\arabic{section}}   
\renewcommand{\thetable}{S\arabic{table}}   
\renewcommand{\thefigure}{S\arabic{figure}}
\renewcommand{\thealgorithm}{S\arabic{algorithm}}

\begin{bibunit}
\begin{center}
{\Large Supplementary material for ``Robust Bayesian Predictive Model Selection Using Bregman Divergence''}
\end{center}

\thispagestyle{empty}

\section{Bregman centroid and information identity}\label{supp:bd}
Bregman divergences have a mean property that is important for our predictive construction. Suppose a collection of densities $p_1, \ldots, p_K$ is summarized by a single representative density. For an arbitrary divergence, the linear mixture density $\bar p = \sum_k w_kp_k$ need not be the Bayes action. For Bregman divergences, however, the ordinary mixture is exactly the right Bregman centroid: it minimizes the weighted average Bregman divergence from the component densities to a single representative density. This identity is what allows the ordinary posterior predictive mixture to retain a Bayes-action interpretation under the same Bregman divergence used for updating and evaluation.

\cref{prop:bd_identity} is the density-level form of the Bregman information identity. \citet{banerjee2005clustering} use the corresponding finite-dimensional identity to define Bregman information and to develop centroid-based Bregman geometry. \citet{Frigyik2008} extend BD to functional settings, including functions and distributions. More recently, \citet{Chodrow2025} shows that the agreement between Jensen-gap information and divergence-from-centroid information characterizes BD: they are the divergences for which these two notions of information agree for arbitrary weighted collections.

\begin{proposition}\label{prop:bd_identity}
    Let $p_1,\ldots,p_J$ be densities with respect to $\lambda$, and let $w_1,\ldots,w_J \ge 0$ satisfy $\sum_{j=1}^J w_j = 1$. Define the mixture density $\bar p(x) = \sum_{j=1}^J w_jp_j(x)$. Assume that the divergences below are finite for all $j=1,\ldots,J$. Then for any density $r$,
    \begin{equation*}
        \sum_{j=1}^J w_j D^\phi(p_j \| r) = \sum_{j=1}^J w_j D^\phi(p_j \| \bar p) + D^\phi(\bar p \| r).
    \end{equation*}
    Consequently,
    \begin{equation*}
        \bar p \in \argmin_r \sum_{j=1}^J w_j D^\phi(p_j \| r).
    \end{equation*}
    If $\phi$ is strictly convex, the minimizer is unique up to equality as a density. Moreover, 
    \begin{equation*}
        \sum_{j=1}^J w_j D^\phi(p_j \| \bar p) = \sum_{j=1}^J w_j \int_\cX \phi(p_j(x)) \lambda (dx) - \int_\cX \phi(\bar p(x)) \lambda (dx).
    \end{equation*}
    Thus, the mean Bregman divergence from the densities to their centroid equals the Jensen-gap of the convex functional $p \mapsto \int_\cX \phi(p(x)) \lambda(dx)$, i.e., the difference between the weighted average of this functional and its value at the mixture density $\bar p$.
\end{proposition}

\begin{proof}
    By the definition of $D^\phi$, we have
    \begin{align}
        &D^\phi(p_j \| r) - D^\phi(p_j \| \bar p) \\
        &= \int_\cX \Big( \phi(\bar p(x)) - \phi(r(x)) - \phi'(r(x)) \big(p_j(x) - r(x) \big) + \phi'(\bar p(x)) \big(p_j(x) - \bar p(x) \big) \Big) \lambda(dx). \nonumber
    \end{align}
    Multiply by $w_j$ and sum over $j$. Since $\bar p(x) = \sum_{j=1}^J w_jp_j(x)$, we have $\sum_{j=1}^J w_j \big(p_j(x) - \bar p(x) \big) =0$. Therefore, 
    \begin{align}
        \sum_{j=1}^J w_j \Big(D^\phi(p_j \| r) - D^\phi(p_j \| \bar p) \Big) &= \int_\cX \Big( \phi(\bar p(x)) - \phi(r(x)) - \phi'(r(x)) \big(\bar p(x) - r(x) \big) \Big) \lambda(dx)  \nonumber \\
        & = D^\phi(\bar p \| r).
    \end{align}
    Rearranging yields 
    \begin{equation}
        \sum_{j=1}^J w_j D^\phi(p_j \| r) = \sum_{j=1}^J w_j D^\phi(p_j \| \bar p) + D^\phi(\bar p \| r).
    \end{equation}
    Since $D^\phi(\bar p \| r) \ge 0$, $\bar p \in \argmin_r \sum_{j=1}^J w_j D^\phi(p_j \| r)$. Strict convexity of $\phi$ gives uniqueness as a density. 

    For the Jensen-gap identity, use the definition of $D^\phi$ again:
    \begin{equation}
        \sum_{j=1}^J w_j D^\phi(p_j \| \bar p) = \sum_{j=1}^J w_j \int_\cX \Big( \phi(p_j(x)) - \phi(\bar p(x)) - \phi'(\bar p(x)) \big(p_j(x) - \bar p(x) \big) \Big) \lambda (dx).
    \end{equation}
    The linear term vanishes because $\sum_{j=1}^J w_j \big(p_j(x) - \bar p(x) \big) =0$. Therefore,
    \begin{equation}
        \sum_{j=1}^J w_j D^\phi(p_j \| \bar p) = \int_\cX \bigg(\sum_{j=1}^J w_j \phi(p_j(x)) - \phi(\bar p(x)) \bigg) \lambda (dx),
    \end{equation}
    which is the stated Jensen gap. This completes the proof. 
\end{proof}

\paragraph{Proof of \cref{prop:breg_cent}}\label{supp:proof:breg_cent}
\begin{proof}
    Write $p_{k,n}^\phi$ simply as $p_{k,n}^\phi$. Using the definition of $D^\phi$, for any density $r$,
    \begin{align}\label{D_diff_1}
        & D^\phi \big(f_k(\cdot; \theta_k) \| r \big) - D^\phi \big(f_k(\cdot; \theta_k) \| p_{k,n}^\phi \big) \\
        &= \int_\cX \Big( \phi(p_{k,n}^\phi(x)) - \phi(r(x)) - \phi'(r(x)) \big(f_k(x; \theta_k) - r(x) \big) + \phi'(p_{k,n}^\phi(x)) \big(f_k(x;\theta_k) -p_{k,n}^\phi(x) \big)  \Big) \lambda(dx). \nonumber
    \end{align}
    Integrate \cref{D_diff_1} with respect to $\pi_{k,n}^\phi(\theta_k \mid x_{1:n}) d\theta_k$. By definition of $p_{k,n}^\phi$, we have $\int_{\Theta_k} \big(f_k(x; \theta_k) - p_{k,n}^\phi(x)\big) \pi_{k,n}^\phi(\theta_k \mid x_{1:n}) d\theta_k = 0$. Hence, after posterior averaging,
    \begin{align}
        &\int_{\Theta_k} \Big(D^\phi \big(f_k(\cdot; \theta_k) \| r \big) - D^\phi \big(f_k(\cdot; \theta_k) \| p_{k,n}^\phi \big) \Big) \pi_{k,n}^\phi(\theta_k \mid x_{1:n}) d\theta_k \\
        & =\int_\cX \Big( \phi(p_{k,n}^\phi(x)) - \phi(r(x)) - \phi'(r(x)) \big(p_{k,n}^\phi(x) - r(x) \big) \Big) \lambda(dx) \\
        &= D^\phi(p_{k,n}^\phi \| r ).
    \end{align}
    Rearranging gives the claim. Since $D^\phi(p_{k,n}^\phi \| r ) \ge 0$, the posterior average Bregman divergence is minimized at $r = p_{k,n}^\phi$. Strict convexity gives uniqueness as a predictive density. The LOO version follows by the same argument with $x_{-i}$ in place of $x_{1:n}$.
\end{proof}

\section{Additional details on the $\beta$-divergence}\label{supp:proof:meth:beta}

In this section, we first give a complementary decision-theoretic interpretation of the $\beta$ posterior predictive used in g-ELPD. Later, we prove the results given in \cref{sec:meth:beta}.

\begin{corollary}\label{cor:beta_jensen}
    Let $\phi_\beta(u) = u^\beta / \big( \beta (\beta-1) \big)$ with $1 \le \beta \le 2$, with $\beta=1$ defined by the log-score limit. Fix $M_k$ and let $p_{k,n}^\beta(\cdot \mid x_{1:n})$ be the $\beta$ posterior predictive density. Then,
    \begin{equation*}
        D^\beta \big( g \, \| \, p_{k,n}^\beta (\cdot \mid x_{1:n}) \big) \le \int_{\Theta_k} D^\beta \big( g \, \| \, f_k(\cdot; \theta_k) \big) \pi_{k,n}^\beta (\theta_k \mid x_{1:n}) d\theta_k.
    \end{equation*}
    Equivalently, the population expected $\beta$ score of $p_{k,n}^\beta(\cdot \mid x_{1:n})$ is at least as large as the posterior average of the population expected $\beta$ scores of $f_k(\cdot; \theta_k)$. The same statement holds for leave-one-out posterior predictive $p_{k,n-1}^\beta (\cdot \mid x_{-i})$. 
\end{corollary}

\paragraph{Proof of \cref{cor:beta_jensen}}\label{supp:proof:beta_jensen}
\begin{proof}
    We prove the result for $1 < \beta \le 2$. The case $\beta=1$ follows by the usual KL limit. For fixed $g(x)$, define
    \begin{equation}
        a(p) = \frac{p^\beta}{\beta} - \frac{g p^{\beta-1}}{\beta-1}.
    \end{equation}
    Since $a$ is continuous at $p=0$ for $\beta >1$, this convexity extends to $[0,\infty)$ by continuity. Then 
    \begin{equation}
        a''(p) = p^{\beta-3} \big((\beta-1)p + (2-\beta)g \big).
    \end{equation}
    For $1 < \beta \le 2$, $p>0$, and $g \ge 0$, this second derivative is nonnegative. Hence $p \mapsto D^\beta(g \| p)$ is convex in its second argument. By Jensen's inequality, we obtain
    \begin{equation}
        D^\beta \bigg( g \| \int_{\Theta_k} f_k(\cdot; \theta_k) \pi_{k,n}^\beta(\theta_k \mid x_{1:n}) d\theta_k \bigg) \le \int_{\Theta_k} D^\beta \big(g \| f_k(\cdot; \theta_k) \big) \pi_{k,n}^\beta(\theta_k \mid x_{1:n}) d\theta_k.
    \end{equation}
    This proves the inequality.
\end{proof}

\cref{cor:beta_jensen} clarifies the role of posterior averaging in the $\beta$ case: for $1\le\beta\le2$, posterior averaging does not increase the population $\beta$-divergence relative to the posterior average plug-in divergence. Thus, the posterior predictive density evaluated by g-ELPD preserves the divergence-based objective that defines the predictive comparison.

\cref{prop:beta_robust} shows that for $\beta>1$, the contribution of any single observation to a pairwise $\beta$-score comparison is uniformly bounded whenever the two predictive densities are bounded. We now provide the proof for this result. 

\paragraph{Proof of \cref{prop:beta_robust}}\label{supp:proof:beta_robust}
\begin{proof}
    For $1 < \beta < \infty$, we have
    \begin{align}
        \Delta_{p,q}^\beta(x) &= S^\beta(p,x) - S^\beta(q,x) \\
        &= \frac{p(x)^{\beta-1} - q(x)^{\beta-1}}{\beta-1} - \frac{1}{\beta} \int_{\cX} \big(p(t)^\beta - q(t)^\beta\big) \lambda(dt).
    \end{align}
    Taking absolute values and applying the triangle inequality gives
    \begin{equation}
        \big \lvert \Delta_{p,q}^\beta(x) \big\rvert \le \frac{\big\lvert p(x)^{\beta-1} - q(x)^{\beta-1} \big\rvert}{\beta-1} + \frac{1}{\beta} \bigg\lvert \int_{\cX} \big(p(t)^\beta - q(t)^\beta\big) \lambda(dt) \bigg\rvert.
    \end{equation}
    Since $0 \le p(x), q(x) \le B$, 
    \begin{equation}
        \big\lvert p(x)^{\beta-1} - q(x)^{\beta-1} \big\rvert \le p(x)^{\beta-1} + q(x)^{\beta-1} \le 2 B^{\beta-1}.
    \end{equation}
    Also,
    \begin{equation}
        \frac{1}{\beta} \bigg\lvert \int_{\cX} \big(p(t)^\beta - q(t)^\beta\big) \lambda(dt) \bigg\rvert \le \frac{1}{\beta} \int_\cX \big\lvert p(t)^\beta - q(t)^\beta\big\rvert \lambda(dt).
    \end{equation}
    Combining these inequalities yields the result.
\end{proof}

We next provide the proof of \cref{cor:contam_bound}, which shows that $\epsilon$-contamination can change the population pairwise $\beta$-score comparison by at most order $\epsilon$.

\paragraph{Proof of \cref{cor:contam_bound}}\label{supp:proof:contam_bound}
\begin{proof}
    Write $\Delta(X) := \Delta_{p,q}^\beta (X)$. 
    By the definition of $g_\epsilon = (1-\epsilon)g + \epsilon h$, we have
    \begin{equation}
        \E_{g_\epsilon}\big( \Delta(X) \big) = (1-\epsilon) \E_{g}\big(\Delta(X) \big) + \epsilon \, \E_h \big(\Delta(X) \big).
    \end{equation}
    Therefore,
    \begin{equation}
        \E_{g_\epsilon}\big( \Delta (X) \big) - \E_{g}\big(\Delta (X) \big) = 
        \epsilon \, \Big( \E_h \big(\Delta (X) \big) - \E_g \big(\Delta (X) \big)  \Big).
    \end{equation}
    Taking absolute values and applying the triangle inequality, 
    \begin{align}
        \Big \lvert \E_{g_\epsilon}\big( \Delta (X) \big) - \E_{g}\big(\Delta (X) \big) \Big \rvert &= \epsilon \Big\lvert \E_h \big(\Delta(X)\big) - \E_g \big(\Delta(X)\big) \Big\rvert\\
        &\le 
        \epsilon \, \Big( \Big\lvert \E_h \big(\Delta (X) \big) \Big \rvert + \Big \lvert \E_g \big(\Delta (X) \big) \Big \rvert   \Big).
    \end{align}
    By \cref{prop:beta_robust}, $ \big\lvert \Delta (X) \big\rvert \le C^\beta(p,q)$ for every $x$. Hence, $\Big\lvert \E_h \big(\Delta (X) \big) \Big \rvert \le \E_h\big\lvert \Delta(X)\big\rvert \le C_\beta(p,q)$, and similarly, $\Big\lvert \E_g \big(\Delta (X) \big) \Big \rvert \le C_\beta(p,q)$. Combining these gives
    \begin{equation*}
        \Big \lvert \E_{g_\epsilon}\big( \Delta (X) \big) - \E_{g}\big(\Delta (X) \big) \Big \rvert \le 2 \epsilon C_\beta(p,q).
    \end{equation*}
    Finally, if $\E_g \big(\Delta (X) \big) > 2\epsilon C_\beta(p,q)$, then 
    \begin{equation*}
        \E_{g_\epsilon}\big( \Delta (X) \big) \ge \E_{g}\big(\Delta (X) \big) - \Big\lvert \E_{g_\epsilon} \big(\Delta(X)\big) - \E_g \big(\Delta(X)\big) \Big\rvert > 0.
    \end{equation*}
    This proves the corollary.
\end{proof}

\section{Numerical Example: Negative-Binomial versus Poisson Regression}\label{supp:sim}
In this simulation, we illustrate the proposed criterion targets a score-dependent predictive pseudo-truth under misspecification. We generate data from a negative binomial regression model with three predictors,
\begin{equation*}
    Y_i \sim \mathrm{NegBin}(\mu_i^*, k^*), \qquad \mu_i^* = \exp(\gamma_0^* + \gamma_1^*X_{i1} + \gamma_2^*X_{i2} + \gamma_3^*X_{i3}),
\end{equation*}
where $k^* = 1$ is the true dispersion parameter, $\mu_i^*$ is the true mean, and $\gamma^* = (\gamma_1^*, \gamma_2^*, \gamma_3^*) = (0.5, -1.0, 1.5)$ are the true coefficients. We compare two misspecified models. The first is a negative binomial regression that omits $X_3$,
\begin{equation*}
    Y_i \sim \mathrm{NegBin}( \mu_{1i},1), \qquad \mu_{1i}=\exp(\gamma_{10}+\gamma_{11}X_{i1}+\gamma_{12}X_{i2}),
\end{equation*}
and the second is a Poisson regression that includes all three predictors,
\begin{equation*}
    Y_i \sim \mathrm{Pois}(\lambda_{2i}), \qquad \lambda_{2i}=\exp(\gamma_{20}+\gamma_{21}X_{i1} +\gamma_{22}X_{i2}+\gamma_{23}X_{i3}).
\end{equation*}
Hence, neither candidate model is correctly specified. The negative binomial model can accommodate overdispersion and heavier tails, but it misses part of the mean structure by omitting $X_3$. The Poisson model includes the full mean structure, but it cannot represent overdispersion because its variance is tied to its mean. 

For each $\beta$, we update each model using the corresponding $\beta$-score generalized posterior and evaluate the same $\beta$-score LOO utility. The posterior computation is implemented in Stan \citep{carpenter2017stan}, and g-ELPD is approximated using the PSIS approximation procedure described in \cref{sec:comp}. We consider sample sizes $n \in \{50,200,500,2000\}$ and $\beta \in \{1, 1.2, 1.4\}$. For each configuration, we repeat the experiment $100$ times and select the model with the larger estimated g-ELPD. 

\begin{figure}
    \centering
    \includegraphics[width=0.6\linewidth]{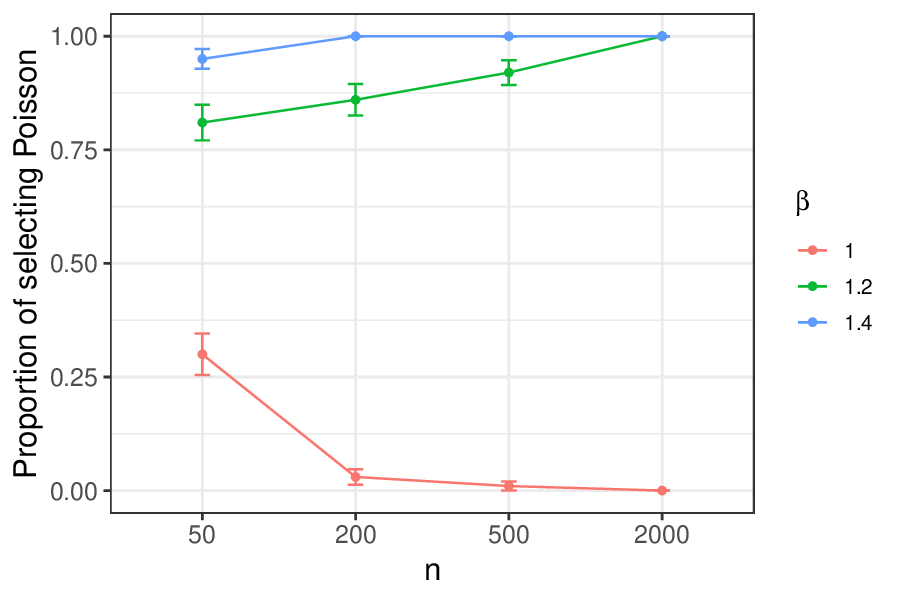}
    \captionsetup{font=small}
    \caption{$\beta$-score model selection under misspecification. Both candidate models are misspecified: the negative-binomial omits $X_3$, whereas the Poisson model includes all predictors but cannot represent overdispersion. Points show the proportion of replications in which the Poisson model has larger g-ELPD. Error bars are approximate binomial standard errors.}
    \label{fig:nb_pois}
\end{figure}
\cref{fig:nb_pois} shows the proportion of replications in which the Poisson model is selected. Under the standard log score ($\beta = 1.0$), the Poisson selection proportion decreases as $n$ grows. This reflects the log score's sensitivity to tail and dispersion mismatch: with larger samples, the overdispersion in the data becomes increasingly visible, favoring the negative binomial candidate despite its omitted covariate. In contrast, as $\beta$ increases to $1.2$ and $1.4$, tail discrepancies receive less weight in the predictive comparison. The selection target therefore shifts toward the model that better captures the covariate-driven mean structure, and the Poisson model is selected more often. This shows that different scoring rules can select different predictive pseudo-truths under misspecification.

\section{Proofs of \cref{sec:theory}}\label{supp:proofs}
In this section, we prove the results in \cref{sec:theory} under \cref{condition1}. Recall the notation in \cref{sec:pre:notation}. Our arguments are carried out for the Bregman score in \cref{def:bregman_score}, so the verification steps differ from the usual log-likelihood setting: in particular, the score contains the integral term $C(\theta)$ and the pointwise term $\phi'(f(x;\theta))$, and we need to control both uniformly in $\theta$.

We first obtain key technical properties of the Bregman score: it is finite, uniformly bounded in $x$ and $\theta$, and continuous in $\theta$.

\begin{lemma}\label{supp:proof:lemma0}(Uniform boundedness of the score)
    Assume \cref{condition1}(b)--(c). Let $C(\theta) = \int_\cX m(f(t; \theta)) \lambda(dt)$, where $m(u) := \phi'(u)u - \phi(u)$. Recall the definition of Bregman score $S^\phi$ in \cref{def:bregman_score}. Then,
    \begin{enumerate}
        \item[(a)] For every $\theta \in \Theta$, $C(\theta)$ is absolutely integrable and finite, and $\big|C(\theta)\big| \le L_m$, where $L_m$ is a Lipschitz constant of $m$ on $[0, C_f]$.
        
        \item[(b)] There exists a finite constant $B < \infty$ such that for all $\theta \in \Theta$ and all $x \in \cX$, $\big| S^\phi\big( f(\cdot; \theta), x\big) \big| \le B$. Consequently, $\big| h_n(\theta) \big| \le B$ for all $n, \theta$.

        \item[(c)] If, in addition, \cref{condition1}(d) holds, then $C: \Theta \to \bbR$ is continuous. Hence $C$ is uniformly continuous on $\Theta$. Moreover, for each fixed $x \in \cX$, the map $\theta \mapsto S^\phi \big( f(\cdot; \theta), x\big)$ is continuous on $\Theta$.
    \end{enumerate}
\end{lemma}

\begin{proof}[Proof of \cref{supp:proof:lemma0}]
\begin{itemize}
    \item[(a)] By \cref{condition1}(c), $m$ is Lipschitz on $[0, C_f]$, i.e., there exists $L_m < \infty$ such that $\big| m(u) - m(v) \big| \le L_m |u-v|$ for all $u,v \in [0, C_f]$. Also, $\phi$ is continuously differentiable on $[0, C_f]$ and $\phi(0)=0$, so $\phi'$ exists at $0$ and is finite, and hence $m(0) = \phi'(0) \cdot 0 - \phi(0) = 0$. Thus, for any $u \in [0, C_f]$, 
    \begin{equation}
        \big| m(u) \big| = \big|m(u) - m(0) \big| \le L_m u.
    \end{equation}
    Fix any $\theta \in \Theta$. By \cref{condition1}(b), $0 \le f(t; \theta) \le C_f$ for all $t$, so the bound applies with $u = f(t; \theta)$. Using that $f(\cdot; \theta)$ is a density with respect to $\lambda$,
    \begin{equation}\label{eq:m_finite}
        \int_\cX \big| m(f(t;\theta)) \big| \lambda(dt) \le L_m \int_\cX f(t;\theta) \lambda(dt) = L_m < \infty.
    \end{equation}
    Thus, the integral defining $C(\theta)$ is absolutely integrable, hence well-defined and finite. Moreover, 
    \begin{equation}
        \big| C(\theta) \big| \le \int_\cX \big| m(f(t;\theta)) \big| \lambda(dt) \le L_m.
    \end{equation}
    This proves part (a).
    
    \item[(b)] By \cref{def:bregman_score}, the Bregman score can be written as $S^\phi\big( f(\cdot; \theta), x\big) = \phi'(f(x;\theta)) - C(\theta)$. By \cref{condition1}(b), $f(x; \theta) \in [0, C_f]$ for all $x$. Also, by \cref{condition1}(c), $\phi'$ is continuous on $[0, C_f]$, and since $[0, C_f]$ is compact, $\phi'$ is bounded there, i.e., $\sup_{u \in [0,C_f]} \big| \phi'(u) \big| < \infty$. Meanwhile, by (a), we have $\big|C(\theta)\big| \le L_m$. Therefore, for all $\theta$ and $x$,
    \begin{equation}
        \big| S^{\phi}\big(f(\cdot; \theta), x\big) \big| = \big| \phi'(f(x;\theta)) - C(\theta) \big| \le \sup_{u \in [0,C_f]} \big| \phi'(u) \big| + L_m < \infty. 
    \end{equation}
    So the claimed uniform bound holds with $B := \sup_{u \in [0,C_f]} \big| \phi'(u) \big| + L_m$. Finally, we obtain
    \begin{equation}
        \big| h_n(\theta) \big| \le \frac{1}{n} \sum_{i=1}^n \big| S^\phi\big(f(\cdot; \theta), X_i \big) \big| \le B.
    \end{equation}
    This proves part (b).    

    \item[(c)] Take any sequence $\theta_j \to \theta$ in $\Theta$. By \cref{condition1}(d), $\sup_{x \in \cX} \big| f(x; \theta_j) - f(x; \theta) \big| \to 0$, hence for every $t \in \cX$,  $f(t; \theta_j) \to f(t; \theta)$. Since each $f(\cdot; \theta_j)$ and $f(\cdot; \theta)$ is a density, Scheffé's lemma (\citealt{williams1991probability}, 5.10) implies 
    \begin{equation}
        \int_\cX \big| f(t; \theta_j) - f(t; \theta) \big| \lambda(dt) \to 0.
    \end{equation}
    Using the Lipschitz property of $m$,
    \begin{align}
        \big| C(\theta_j) - C(\theta) \big| & \le \int_\cX \big| m(f(t; \theta_j)) - m(f(t;\theta)) \big| \lambda(dt) \\
                                            & \le L_m \int_\cX \big| f(t; \theta_j) - f(t; \theta) \big| \lambda(dt) \to 0.
    \end{align}
    Thus $C$ is continuous on $\Theta$. Since $\Theta$ is compact, $C$ is uniformly continuous on $\Theta$.

    Finally, fix $x \in \cX$. Since $f(x; \theta_j) \to f(x; \theta)$ and $\phi'$ is continuous on $[0, C_f]$ (by \cref{condition1}(c)), we have $\phi'(f(x;\theta_j)) \to \phi'(f(x;\theta))$. Together with $C(\theta_j) \to C(\theta)$, this yields $S^\phi \big( f(\cdot; \theta_j), x\big) \to S^\phi \big( f(\cdot; \theta), x\big)$. Hence $\theta \mapsto S^\phi \big( f(\cdot; \theta), x\big)$ is continuous. This completes the proof. 
\end{itemize}
\end{proof}

\cref{supp:proof:lemma1} provides the equicontinuity of the empirical average score $h_n(\theta)$ defined in \cref{condition1}.

\begin{lemma}\label{supp:proof:lemma1}(Equicontinuity of sample score)
    Assume \cref{condition1} (b)--(d). Then, $\{h_n\}_{n \ge 1}$ is equicontinuous on $\Theta$; i.e., for every $\epsilon >0$ there exists $\delta>0$ such that for all $n \in \bbN$ and all $\theta, \theta' \in \Theta$, if $\|\theta - \theta'\| < \delta$ then $\big|h_n(\theta) - h_n(\theta')\big| < \epsilon$.
\end{lemma}

\begin{proof}[Proof of \cref{supp:proof:lemma1}]
    Recall from \cref{def:bregman_score} that the Bregman scoring rule can be written in the form $S^\phi\big( f(\cdot; \theta), x \big) = \phi'(f(x; \theta)) - C(\theta)$, where 
    \begin{equation}
        C(\theta) := \int_\cX \bigg( \phi'\big( f(t; \theta) \big) f(t; \theta) - \phi\big( f(t; \theta) \big) \bigg) \lambda (dt) = \int_\cX m\big(f(t;\theta)\big) \lambda(dt).
    \end{equation}
    Here, $m(u) = \phi'(u)u - \phi(u)$. Fix $\epsilon >0$. For any $\theta, \theta' \in \Theta$,
    \begin{equation}
        h_n(\theta) - h_n(\theta') = \frac{1}{n}\sum_{i=1}^n \bigg( \phi'\big( f(X_i; \theta) \big) - \phi'\big( f(X_i; \theta') \big) \bigg) - \big( C(\theta) - C(\theta') \big).
    \end{equation}
    By the triangle inequality,
    \begin{equation}\label{eq:hn_control}
        \big| h_n(\theta) - h_n(\theta') \big| \le \frac{1}{n}\sum_{i=1}^n \Big| \phi'\big( f(X_i; \theta) \big) - \phi'\big( f(X_i; \theta') \big) \Big| + \big| C(\theta) - C(\theta') \big|.
    \end{equation}
    So it suffices to bound the sample-average term and the $C(\cdot)$ term by $\epsilon/2$ each, uniformly in $n$.
    
    We first control the sample term in \cref{eq:hn_control}. By \cref{condition1}(b), $0 \le f(x; \theta) \le C_f$ for all $(x,\theta) \in \cX \times \Theta$. By \cref{condition1}(c), $\phi'$ is continuous on $[0,C_f]$, hence uniformly continuous on $[0,C_f]$. Therefore, there exists $\eta >0$ such that whenever $u,v \in [0, C_f]$ and $|u-v| < \eta$, one has $\big| \phi'(u) - \phi'(v) \big| < \epsilon/2$. By \cref{condition1}(d), there exists $\delta_1 >0$ such that $\|\theta-\theta'\| < \delta_1$ implies $\sup_{x \in \cX} \big| f(x;\theta) - f(x;\theta') \big| < \eta$. In particular, for each $i$, $\big| f(X_i; \theta) - f(X_i; \theta') \big| < \eta$, hence $\big| \phi'\big( f(X_i; \theta) \big) - \phi'\big( f(X_i; \theta') \big) \big| < \epsilon/2$. Averaging over $i=1,\ldots,n$ yields $\tfrac{1}{n}\sum_{i=1}^n \big| \phi'\big( f(X_i; \theta) \big) - \phi'\big( f(X_i; \theta') \big) \big| < \epsilon/2$ for all $n$ whenever $\|\theta-\theta'\| < \delta_1$.
    
    Next, consider the integral term in \cref{eq:hn_control}. By \cref{supp:proof:lemma0}(a), $C(\theta)$ is well-defined and finite for each $\theta$, and by \cref{supp:proof:lemma0}(c), $C$ is uniformly continuous on $\Theta$. Therefore, there exists $\delta_2 >0$ such that $\|\theta- \theta'\| < \delta_2$ implies $\big| C(\theta) - C(\theta') \big| < \epsilon/2$.
    
    Let $\delta := \min(\delta_1, \delta_2)$. If $\|\theta - \theta'\| < \delta$, then for all $n$, $\big|h_n(\theta) - h_n(\theta')\big| \le \epsilon/2 + \epsilon/2 = \epsilon$. This establishes equicontinuity of $\{h_n\}_{n \ge 1}$ on $\Theta$.
\end{proof}

Before proving the concentration result, \cref{supp:proof:lemma2} verifies that the Bregman posterior in \cref{eq:breg_post} is a valid probability measure.

\begin{lemma}\label{supp:proof:lemma2}(Bregman posterior is well-defined) 
    Assume \cref{condition1}(b)--(c). Let $\pi$ be any prior probability measure on $(\Theta, \mathcal{B}(\Theta))$ where $\mathcal{B}(\Theta)$ is the Borel $\sigma$-algebra on $\Theta$. Then for every $n \in \bbN$, the normalizing constant
    \[
        Z_n := \int_\Theta \exp \big( n h_n(\theta) \big) \pi(d\theta)
    \]
    is finite, and hence the Bregman posterior $\pi_n^\phi(\cdot \mid X_{1:n})$ defined by 
    \begin{equation}\label{eq:breg_post_full}
        \pi_n^\phi(d\theta \mid X_{1:n}) = \frac{\exp \big( n h_n(\theta) \big) \pi(d\theta)}{Z_n}
    \end{equation}
    is well-defined probability measure on $(\Theta, \mathcal{B}(\Theta))$.
\end{lemma}

\begin{proof}[Proof of \cref{supp:proof:lemma2}]
    By \cref{supp:proof:lemma0}(b), there exists $B < \infty$ such that for all $\theta \in \Theta$, $\big|h_n(\theta)\big| \le B$. In particular, $h_n(\theta) \le B$, so for every $\theta$, $\exp\big(nh_n(\theta)\big) \le \exp(nB)$. Integrating both sides with respect to the prior $\pi$ yields
    \begin{equation}
        Z_n = \int_\Theta \exp \big( n h_n(\theta) \big) \pi(d\theta) \le \int_\Theta \exp(nB) \pi(d\theta) = \exp(nB) < \infty,
    \end{equation}
    since $\pi$ is a probability measure. Thus $Z_n$ is finite. Moreover, $\exp\big(nh_n(\theta)\big) >0$ for all $\theta$, so $Z_n >0$. Therefore, $\pi_n^\phi(d\theta \mid X_{1:n})$ is a well-defined probability measure on $(\Theta, \mathcal{B}(\Theta))$. 
\end{proof}

\cref{supp:proof:lemma3} shows that $h_n$ satisfies a uniform strong law of large numbers on $\Theta$. This uniform convergence is required when we apply Theorem 3 of \citet{miller2021asymptotic}.

\begin{lemma}\label{supp:proof:lemma3}(Uniform SLLN for the sample score)
    Assume \cref{condition1}(b)--(d). Then 
    \[
        \sup_{\theta \in \Theta} \big| h_n(\theta) - h(\theta) \big| \xrightarrow[n \to \infty]{\mathrm{a.s.}} 0.
    \]
\end{lemma}

\begin{proof}[Proof of \cref{supp:proof:lemma3}]
    Fix any $\theta \in \Theta$. By \cref{supp:proof:lemma0}(b), there exists $B < \infty$ such that $\big|S^\phi\big(f(\cdot; \theta), x\big)\big| \le B$ for all $x \in \cX$. Hence $\E \big(\big|S^\phi\big(f(\cdot; \theta), X\big)\big|\big) \le B < \infty$, and the strong law of large numbers (SLLN) yields $h_n(\theta) \to h(\theta)$ almost surely. 

    We then show the continuity of $h$ on $\Theta$. Let $\theta_m \to \theta$. By \cref{supp:proof:lemma0}(c), for each fixed $x \in \cX$, the map $\theta \mapsto S^\phi\big(f(\cdot; \theta),x\big)$ is continuous. Together with the uniform bound $\big|S^\phi\big(f(\cdot; \theta_m), X\big)\big| \le B$ for all $m$, the dominated convergence theorem gives $h(\theta_m) \to h(\theta)$. Thus, $h$ is continuous on $\Theta$. Since $\Theta$ is compact, $h$ is uniformly continuous on $\Theta$.

    We now show the uniform convergence on $\Theta$, following the argument in Lemma 36 of \citet{miller2021asymptotic}. Fix $\epsilon>0$. By \cref{supp:proof:lemma1}, $\{h_n\}_{n\ge1}$ is equicontinuous on $\Theta$. Hence there exists $\delta_1 >0$ such that for all $n \ge 1$ and all $\theta,\theta' \in \Theta$, if $\|\theta-\theta'\|<\delta_1$ then $|h_n(\theta) - h_n(\theta')| < \epsilon$. By uniform continuity of $h$, there exists $\delta_2 >0$ such that for all $\theta,\theta' \in \Theta$, if $\|\theta-\theta'\|<\delta_2$ then $|h(\theta) - h(\theta')| < \epsilon$. 

    Let $\delta := \min\{\delta_1, \delta_2\}$. Since $\Theta$ is totally bounded (in particular, any compact subset of $\bbR^d$ is totally bounded), there exist $\theta_1, \ldots, \theta_J \in \Theta$ such that for every $\theta \in \Theta$ there is some $j \in \{1,\ldots,J\}$ with $\|\theta-\theta_j\| < \delta$. Since we have pointwise convergence, for each fixed $j$, $h_n(\theta_j) \to h(\theta_j)$ almost surely. Because there are only finitely many $\theta_j$'s, we also have 
    \begin{equation}
        \max_{1\le j\le J} \big| h_n(\theta_j) - h(\theta_j) \big| \to 0 \quad \mathrm{a.s.}
    \end{equation}
    Therefore, there exists an event $\Omega_\epsilon$ with $P(\Omega_\epsilon) =1$ such that for every $\omega \in \Omega_\epsilon$ there exists $N_\epsilon(\omega)$ with 
    \begin{equation}
        \max_{1\le j\le J} \big| h_n(\theta_j) - h(\theta_j) \big| < \epsilon, \quad \forall n \ge N_\epsilon(\omega).
    \end{equation}
    Fix $\omega \in \Omega_\epsilon$ and $n \ge N_\epsilon(\omega)$. For any $\theta \in \Theta$, choose $j$ with $\|\theta-\theta_j\| < \delta$. Then
    \begin{equation}
        \big|h_n(\theta) - h(\theta) \big| \le \big|h_n(\theta) - h_n(\theta_j) \big| + \big|h_n(\theta_j) - h(\theta_j) \big| + \big|h(\theta_j) - h(\theta) \big| < \epsilon + \epsilon + \epsilon = 3\epsilon.
    \end{equation}
    Taking the supremum over $\theta \in \Theta$, we obtain that on $\Omega_\epsilon$,
    \begin{equation}\label{eq:uniform_hn}
        \sup_{\theta \in \Theta}\big|h_n(\theta) - h(\theta)\big| < 3\epsilon, \quad \forall n \ge N_\epsilon(\omega).
    \end{equation}

    Finally, let $\epsilon_m := 1/m$ and define $\Omega_m := \Omega_{\epsilon_m}$. Each $\Omega_m$ satisfies $P(\Omega_m) = 1$, hence the countable intersection $\Omega_0 := \cap_{m=1}^\infty \Omega_m$ still satisfies $P(\Omega_0) = 1$. Fix any $\omega \in \Omega_0$. Let $\eta >0$ be arbitrary and choose $m := \lceil 3/\eta \rceil + 1$. Since $\omega \in \Omega_m$, by \cref{eq:uniform_hn} we have $\sup_{\theta \in \Theta}\big|h_n(\theta) - h(\theta)\big| < 3\epsilon_m < \eta$ for all $n \ge N_{\epsilon_m}(\omega)$. Since $\eta>0$ was arbitrary, this completes the proof. 
\end{proof}

The proof of \cref{thm:posterior_consistency} is an application of \cite{miller2021asymptotic} Theorem 3, and the supporting lemmas above are used to verify its case 2 assumptions in our Bregman score setup.


\begin{proof}[Proof of \cref{thm:posterior_consistency}]
    The Bregman posterior in \cref{eq:breg_post_full} can be written as 
    \begin{equation}\label{eq:breg_post_hn}
        \pi_n^\phi(d\theta \mid X_{1:n}) \propto \pi(d\theta) \, \exp \bigg( \sum_{i=1}^n S^{\phi}\big(f(\cdot;\theta), X_i\big) \bigg) = \pi(d\theta) \exp\big( nh_n(\theta) \big).
    \end{equation}
    By \cref{supp:proof:lemma2}, $\pi_n^\phi(d\theta \mid X_{1:n})$ is well-defined for all $n$. We apply Theorem 3 of \citet{miller2021asymptotic} for generalized posteriors on a metric space. Define $f_n(\theta) := - h_n(\theta)$, $f(\theta) := - h(\theta)$, and set $\theta_0 = \theta^*$. Equip $\Theta$ with the metric $d(\theta, \theta') = \|\theta - \theta'\|$. 
    
    \cref{condition1}(e) states that the prior $\pi$ assigns strictly positive mass to every open neighborhood of $\theta^*$. This matches the prior-positivity requirement. By \cref{supp:proof:lemma3}, we have $\sup_{\theta \in \Theta} \big| h_n(\theta) - h(\theta) \big| \to 0$ almost surely, which implies $f_n(\theta) \to f(\theta)$ almost surely for every $\theta \in \Theta$.

    We proceed to verify the assumption case (2) of Theorem 3 in \citet{miller2021asymptotic}. Let $K := \Theta$. By \cref{condition1}, $\Theta$ is compact. Then it suffices to show that $(f_n)$ is equicontinuous and $f(\theta) > f(\theta^*)$ for all $\theta \in \Theta \setminus \{\theta_0\}$. \cref{supp:proof:lemma1} gives that $(h_n)$ is equicontinuous on $\Theta$. Therefore, $(f_n)$ is also equicontinuous on $\Theta$. Next, by \cref{condition1}(a), since $\theta^*$ is the unique maximizer of $h$, for every $\eta >0$, $\inf_{\|\theta - \theta^*\| \ge \eta} f(\theta) > f(\theta^*)$. 

    Thus, all conditions required by Theorem 3 are verified, and the results follows. 
\end{proof}

We now prove \cref{thm:ppc}.

\begin{proof}[Proof of \cref{thm:ppc}]
    By \cref{thm:posterior_consistency}, for every $\delta >0$,
    \begin{equation}
        \pi_n^\phi\big( \{\theta \in \Theta: \|\theta-\theta^*\| \ge \delta \} \mid X_{1:n}\big) \xrightarrow[n \to \infty]{\mathrm{a.s.}} 0.
    \end{equation}
    Fix $\epsilon >0$. By \cref{condition1}(d), there exists $\delta >0$ such that whenever $\|\theta-\theta^*\|<\delta$,
    \begin{equation}
        \sup_{x \in \cX} \big| f(x; \theta) - f(x; \theta^*) \big| < \epsilon.
    \end{equation}
    
    Fix $x \in \cX$. By the triangle inequality,
    \begin{align}
        \big|p_n^\phi(x \mid X_{1:n}) - f(x; \theta^*) \big| &= \Bigg| \int_\Theta  \big( f(x; \theta) - f(x; \theta^*) \big) \pi^\phi_n(d\theta \mid X_{1:n}) \Big| \\
            &\le \int_\Theta \big| f(x; \theta) - f(x; \theta^*) \big| \pi^\phi_n(d\theta \mid X_{1:n}).
    \end{align}
    Split the integral over the sets $\{ \|\theta-\theta^*\| < \delta\}$ and $\{ \|\theta-\theta^*\| \ge \delta\}$:
    \begin{equation}
        \int_\Theta \big| f(x; \theta) - f(x; \theta^*) \big| \pi^\phi_n(d\theta \mid X_{1:n}) \le I_{n,1}(x) + I_{n,2}(x),
    \end{equation}
    where 
    \begin{equation}
        I_{n,1}(x) := \int_{\|\theta - \theta^*\| < \delta} \big| f(x; \theta) - f(x; \theta^*) \big| \pi^\phi_n(d\theta \mid X_{1:n}),
    \end{equation}
    and 
    \begin{equation}
        I_{n,2}(x) := \int_{\|\theta - \theta^*\| \ge \delta} \big| f(x; \theta) - f(x; \theta^*) \big| \pi^\phi_n(d\theta \mid X_{1:n}).
    \end{equation}
    For the first term, by the choice of $\delta$,
    \begin{align}
        I_{n,1}(x) &\le \sup_{\substack{\theta \in \Theta: \\ \|\theta-\theta^*\| < \delta}} \big| f(x; \theta) - f(x; \theta^*) \big| \\
                   &\le \sup_{\substack{\theta \in \Theta: \\ \|\theta-\theta^*\| < \delta}} \sup_{x \in \cX} \big| f(x; \theta) - f(x; \theta^*) \big| < \epsilon.
    \end{align}
    For the second term, \cref{condition1}(b) implies $0 \le f(x; \theta) \le C_f$ for all $x \in \cX$, $\theta \in \Theta$, hence $\big|f(x;\theta) - f(x; \theta^*)\big| \le 2C_f$, so
    \begin{equation}
        I_{n,2}(x) \le 2C_f \,\pi_n^\phi\big( \{\theta: \|\theta-\theta^*\| \ge \delta \} \mid X_{1:n}\big). 
    \end{equation}

    Combining the bounds, for all $n$ and all $x \in \cX$,
    \begin{equation}
        \big|p_n^\phi(x \mid X_{1:n}) - f(x; \theta^*) \big| \le \epsilon + 2C_f \,\pi_n^\phi\big( \{\theta: \|\theta-\theta^*\| \ge \delta \} \mid X_{1:n}\big). 
    \end{equation}
    Taking $\sup_{x \in \cX}$ on both sides gives
    \begin{equation}
        \sup_{x \in \cX} \big|p_n^\phi(x \mid X_{1:n}) - f(x; \theta^*) \big| \le \epsilon + 2C_f \,\pi_n^\phi\big( \{\theta: \|\theta-\theta^*\| \ge \delta \} \mid X_{1:n}\big). 
    \end{equation}
    By \cref{thm:posterior_consistency}, the posterior mass term on the right converges to $0$ almost surely. Therefore,
    \begin{equation}
        \limsup_{n \to \infty} \, \sup_{x \in \cX} \big|p_n^\phi(x \mid X_{1:n}) - f(x; \theta^*) \big|  \le \epsilon, \quad \mathrm{a.s.}
    \end{equation}
    Since $\epsilon>0$ was arbitrary, it follows that 
    \begin{equation}
        \sup_{x \in \cX} \big| p_n^\phi(x \mid X_{1:n}) - f(x; \theta^*) \big| \xrightarrow[n \to \infty]{\mathrm{a.s.}} 0.
    \end{equation}
\end{proof}

We establish a technical lemma to help in proving \cref{supp:proof:loo_consistency}. 
\begin{lemma}\label{supp:proof:lemma4}
    Recall the Bregman score $S^\phi$ in \cref{def:bregman_score}, and assume \cref{condition1}(b)--(c). Let $q$ be a density with $0 \le q(x) \le C_f$, and let $\{q_n\}_{n \ge 1}$ be (possibly random) densities with $0 \le q_n(x) \le C_f$ for all $x$, such that $\sup_{x \in \cX} \big|q_n(x) - q(x)\big| \xrightarrow[]{\mathrm{a.s.}} 0$ as $n \to \infty$. Let $X \sim G$ and define $\Delta_n := S^\phi(q_n, X) - S^\phi(q, X)$. Then, $\Delta_n \xrightarrow[]{\mathrm{a.s.}} 0$ as $n \to \infty$, and $\E|\Delta_n| \to 0$.
\end{lemma}

\begin{proof}[Proof of \cref{supp:proof:lemma4}]
    Let $m(u) := \phi'(u)u-\phi(u)$. Then, the Bregman score can be written as $S^\phi(P,x) = \phi'(p(x)) - \int_\cX m(p(t)) \lambda(dt)$. Write 
    \begin{equation}
        \Delta_n = \underbracket{\phi'(q_n(X)) - \phi'(q(X))}_{=: A_n}  - \underbracket{\int_{\cX} \Big( m(q_{n}(t)) - m(q(t)) \Big) \lambda(dt)}_{=: B_n}.
    \end{equation}
    
    We first control $A_n$. Since $\sup_x \big|q_n(x) - q(x)\big| \to 0$ a.s., we have $\big| q_n(X) - q(X) \big| \to 0$ a.s. By \cref{condition1}(b), $q_{n}(x), q(x) \in [0, C_f]$. By \cref{condition1}(c), $\phi'$ is continuous on $[0, C_f]$, hence uniformly continuous there. Hence, for any $\epsilon >0$, there exists $\delta >0$ such that $|u-v| < \delta$ implies $\big|\phi'(u) - \phi'(v)\big| < \epsilon$, and thus $|A_n| \to 0$ a.s.

    Next, control $B_n$. By \cref{condition1}(c), $m$ is Lipschitz on $[0,C_f]$ with some constant $L_m$. Thus, 
    \begin{equation}
        |B_n| \le \int_{\cX} \big|m(q_{n}(t)) - m(q(t)) \big| \lambda(dt) \le L_m \int_{\cX}|q_{n}(t) - q(t)| \lambda(dt).
    \end{equation}
    Since both $q_n$ and $q$ are densities and $q_{n}(t) \to q(t)$ pointwise for every $t$, Scheffé's lemma (\citealt{williams1991probability}, 5.10) gives 
    \begin{equation}
        \int_{\cX}|q_{n}(t) - q(t)| \lambda(dt) \xrightarrow[n \to \infty]{\mathrm{a.s.}} 0.
    \end{equation}
    Hence $|B_n| \to 0$ a.s. 

    Combining, $\Delta_n = A_n - B_n \xrightarrow[]{\mathrm a.s.} 0$ as $n \to \infty$. Moreover, by the same argument as in the proof of \cref{supp:proof:lemma0}(b), $\big| S^\phi(p,x) \big| \le B$ for all densities $p$ with $0\le p \le C_f$ and all $x \in \cX$, where $B= \sup_{u \in [0,C_f]}|\phi'(u)| + L_m$. Therefore,
    \begin{equation}
        |\Delta_n| = \big|S^\phi(q_{n}, X) - S^\phi(q, X)\big| \le \big|S^\phi(q_{n}, X)\big| + \big|S^\phi(q, X)\big| \le 2B,
    \end{equation}
    and dominated convergence yields $\E|\Delta_n| \to 0$. This completes the proof.    
\end{proof}

To work with leave-one-out predictive criteria, we need to know that randomly deleting observations produces the same distributional structure as taking the first $n-s$ observations and holding out the last $s$ observations. \cref{supp:proof:lemma5} formalizes this fact: if a random subset of indices is deleted independently of data, then the retained and deleted observations have the same joint distribution as the first $n-s$ and last $s$ observations of an i.i.d sample. 

\begin{lemma}\label{supp:proof:lemma5} 
    Fix $s \in \bbN$, and suppose $n > s$. Let $X_1, \ldots, X_n$ be i.i.d. random variables drawn from a probability distribution $g$ on a measurable space $(\cX, \cA)$. Let $I \subseteq \{1, \ldots, n\}$ be a random subset of size $s$, drawn independently of the data. Write $X_{-I}$ and $X_I$ for the retained and deleted observations, respectively, ordered by their original indices. Then $(X_{-I}, X_I)$ has the same distribution as $(X_{1:n-s}, X_{n-s+1:n})$. In particular, $X_{-I}$ consists of $n-s$ i.i.d. draws from $g$, $X_I$ consists of $s$ i.i.d. draws from $g$, and the two blocks are independent. 
\end{lemma}

\begin{proof}[Proof of \cref{supp:proof:lemma5}]
    Because $I$ is drawn independently of data, it suffices to condition on the event $I=A$, where $A$ is any fixed subset of $\{1,\ldots,n\}$ with $|A|=s$. Conditional on $I=A$, the vector $(X_{-A}, X_A)$ is just a reordering of the original i.i.d. sample $(X_1,\ldots, X_n)$. Since an i.i.d. sample has the same joint distribution under any permutation of its coordinates, $(X_{-A}, X_A)$ has the same distribution as $(X_{1:n-s}, X_{n-s+1:n})$. The conditional distribution is the same for every subset $A$ of size $s$. Therefore, averaging over the random choice of $I$ gives the distributional equality. Since the right-hand side consists of two disjoint blocks of an i.i.d. sample, $X_{-I}$ consists of $n-s$ i.i.d. draws from $g$, $X_I$ consists of $s$ i.i.d. draws from $g$, and the two blocks are independent.
\end{proof}

\begin{theorem}\label{supp:proof:loo_consistency}
    Assume \cref{condition1}(a)--(e) holds for model $k$. Let $\hat{u}^{\mathrm{loo}}_{n,k}$ be the LOO utility estimator in \cref{eq:loo_utility}, and let $\bar u_k$ be the limiting population utility in \cref{eq:lim_ut}. Then, $\hat{u}^{\mathrm{loo}}_{n,k} \xrightarrow[n \to \infty]{\mathrm{p}} \bar{u}_k$.
\end{theorem}

\begin{proof}[Proof of \cref{supp:proof:loo_consistency}]
    Let $L_n \sim \mathrm{Uniform}(1, \ldots, n)$, independently of the sample $X_1, \ldots, X_n$. Then, 
    \begin{equation}
        \hat{u}^{loo}_{n,k} = \frac{1}{n}\sum_{i=1}^n S^\phi \big( p_{n-1}^\phi(\cdot \mid X_{-i}), X_i \big) = \E_{L_n} \Big( S^\phi\big(p^\phi_{n-1}(\cdot \mid X_{-L_n}), X_{L_n}\big) \mid X_{1:n} \Big).
    \end{equation}
    We want to show that $\hat{u}^{loo}_{n,k} \xrightarrow[]{\mathrm{p}} \bar{u}_k$ as $n \to \infty$, where $\bar{u}_k := \E \big( S^\phi\big(f(\cdot; \theta^*), X\big) \big)$. 

    Let $q(\cdot) := f(\cdot; \theta^*)$ and $q_{n-1}^{(-L_n)}(\cdot) := p_{n-1}^\phi(\cdot \mid X_{-L_n})$. Then
    \begin{equation}\label{eq:uloo_diff}
        \hat{u}^{loo}_{n,k} - \bar{u}_k = \E_{L_n}\big(S^\phi(q^{(-L_n)}_{n-1},X_{L_n}) - S^\phi(q, X_{L_n}) \mid X_{1:n}\big) + \left( \frac{1}{n} \sum_{i=1}^n S^\phi(q, X_i) - \bar{u}_k \right).
    \end{equation}
    By \cref{supp:proof:lemma0}(b), $S^\phi\big(f(\cdot;\theta),x\big)$ is uniformly bounded, hence $S^\phi(q,X)$ is integrable. Since $\E_{L_n} \big(S^\phi(q,X_{L_n}) \mid X_{1:n}\big) = \tfrac{1}{n}\sum_{i=1}^n S^\phi(q, X_i)$, the second term in \cref{eq:uloo_diff} converges to $0$ almost surely by the SLLN. It remains to show the first term in \cref{eq:uloo_diff} converges in probability to $0$. 
    
    Set $\Delta_n := S^\phi(q^{(-L_n)}_{n-1}, X_{L_n}) - S^\phi(q, X_{L_n})$. By Jensen's inequality, $\big| \E_{L_n}(\Delta_n \mid X_{1:n}) \big| \le \E_{L_n} \big(|\Delta_n| \, \big| X_{1:n}\big)$. Then, for $\epsilon >0$, Markov's inequality gives
    \begin{equation}
        P\big(\big|\E_{L_n}(\Delta_n \mid X_{1:n})\big| > \epsilon \big) \le \frac{1}{\epsilon}\E \big( \big|\E_{L_n}(\Delta_n \mid X_{1:n})\big|\big) \le \frac{1}{\epsilon} \E\big( \E_{L_n} (|\Delta_n| \, \big| X_{1:n}) \big) = \frac{1}{\epsilon}\E|\Delta_n|.
    \end{equation}
    So it suffices to show $\E|\Delta_n| \to 0$.

    Apply \cref{supp:proof:lemma5} with $s=1$ and $I=\{L_n\}$. Then $(X_{-L_n}, X_{L_n})$ has the same distribution as $(X_{1:n-1}, X_n)$. Hence $(q^{(-L_n)}_{n-1}, X_{L_n})$ has the same distribution as $(\tilde{q}_{n-1}, X_n)$, where $\tilde{q}_{n-1}(\cdot) := p_{n-1}^\phi (\cdot \mid X_{1:n-1})$. Therefore, $\E|\Delta_n| = \E\big| S^\phi(\tilde{q}_{n-1}, X_n) - S^\phi(q, X_n) \big|$, so we may work with $\tilde{q}_{n-1}$ and the independent holdout $X_n$. By \cref{thm:ppc} applied to the sample $X_{1:n-1}$ (as $n-1 \to \infty$), $\sup_{x \in \cX} \big|\tilde{q}_{n-1}(x) - q(x) \big| \xrightarrow[]{\mathrm{a.s.}} 0$. Moreover, by \cref{condition1}(b), we have $0 \le f_k(x;\theta) \le C_f$ for all $x$ and $\theta$, hence $0\le q(x)\le C_f$ for all $x$. Since $\tilde q_{n-1}$ is a posterior predictive mixture of such bounded model densities, it also satisfies $0\le \tilde q_{n-1}(x)\le C_f$ for all $x$. Therefore, \cref{supp:proof:lemma4} applies to the sequence $\tilde{q}_{n-1}$ and yields $\E|\Delta_n| \to 0$. This proves the claim.
\end{proof}

Finally, we provide the proof of \cref{sec:theory:consistency}.

\begin{proof}[Proof of Theorem \ref{sec:theory:consistency}]
Fix a model index $k \in \{1,\ldots, K\}$. Under \cref{condition1} applied to model $k$, \cref{supp:proof:loo_consistency} establishes that $\hat{u}_n^{loo}(M_k) \xrightarrow[]{\mathrm{p}} \bar{u}(M_k)$ as $n \to \infty$. Define $\Delta_n := \max_{1 \le k \le K} \big| \hat{u}^{\mathrm{loo}}_{n,k} - \bar{u}_k\big|$. We claim $\Delta_n \xrightarrow[]{\mathrm{p}} 0$. To see this, fix $\epsilon >0$. By a union bound,
\begin{equation}
    P(\Delta_n > \epsilon) = P\bigg(\bigcup_{k=1}^K \big\{ \big| \hat{u}^{\mathrm{loo}}_{n,k} - \bar{u}_k > \epsilon\big| \big\} \bigg) \le \sum_{k=1}^K P\big( \big| \hat{u}^{\mathrm{loo}}_{n,k} - \bar{u}_k\big| > \epsilon \big).
\end{equation}
Since $K < \infty$ and each term on the right converges to $0$, the sum also converges to $0$. Hence, $P(\Delta_n > \epsilon) \to 0$, i.e., $\Delta_n \xrightarrow[]{\mathrm{p}} 0$.

By assumption, the optimal model is unique: there exists a unique $k^*$ such that $\bar{u}_{k^*} > \bar{u}_k$ for all $k \neq k^*$. Define $\gamma := \bar{u}_{k^*} - \max_{k \neq k^*} \bar{u}_k > 0$. Consider the event $E_n := \{\Delta_n < \gamma/2\}$. On $E_n$, for any $k \neq k^*$,
\begin{equation}
    \hat{u}^{\mathrm{loo}}_{n,k^*} \ge \bar{u}_{k^*} - \Delta_n > \bar{u}_{k^*} - \frac{\gamma}{2} = \max_{j \neq k^*} \bar{u}_j + \frac{\gamma}{2} \ge \bar{u}_k + \frac{\gamma}{2} \ge \hat{u}^{\mathrm{loo}}_{n,k}.
\end{equation}
Hence, on $E_n$, the maximizer of $k \mapsto \hat{u}^{\mathrm{loo}}_{n,k}$ is uniquely $k^*$, i.e., $\widehat{M}_n = M^*$. Thus, since $\Delta_n \xrightarrow[]{\mathrm{p}} 0$, 
\begin{equation}
    P(\widehat{M}_n \neq M^*) \le P(E_n^c) \to 0.
\end{equation}
Therefore, $\widehat{M}_n \xrightarrow[]{\mathrm p} M^*$, proving the theorem. 
\end{proof}

\paragraph{Proof of \cref{thm:loo_incoherent}}
\begin{proof}[Proof of \cref{thm:loo_incoherent}]
    Let us prove (a) first. The argument is identical to the proof of \cref{supp:proof:loo_consistency}, except that the posterior predictive is learned under $\phi_2$ while the LOO utility is evaluated under $\phi_1$. Thus, we only record the changes relative to the proof of \cref{supp:proof:loo_consistency}. Specifically, the claim follows from the proof of \cref{supp:proof:loo_consistency} with the substitutions $S^\phi \rightsquigarrow S^{\phi_1}$, $q_{n-1}(\cdot) := p^\phi_{n-1}(\cdot\mid X_{-I}) \rightsquigarrow p^{\phi_2}_{n-1}(\cdot\mid X_{-I})$, $q(\cdot):=f(\cdot;\theta^*) \rightsquigarrow f(\cdot;\theta^{*(2)}_k)$, and with $\bar u_k \rightsquigarrow \bar u^{(1|2)}_k$.
    
    In particular, the proof of \cref{supp:proof:loo_consistency} reduces (via the same Jensen/Markov and exchangeability steps) to showing
    $\E\Big|S^{\phi_1}(\tilde q_{n-1},X_n)-S^{\phi_1}(q,X_n)\Big|\to 0$ where $\tilde q_{n-1}(\cdot) := p^{\phi_2}_{n-1}(\cdot\mid X_{1:n-1})$ and $q(\cdot)=f_k(\cdot;\theta^{*(2)}_k)$. By assumption, \cref{thm:ppc} holds for $\phi_2$, so applied to the sample $X_{1:n-1}$ it yields $\sup_{x\in\cX}|\tilde q_{n-1}(x) - q(x)|\to 0$ almost surely. As in the proof of \cref{supp:proof:loo_consistency}, $0\le q(x)\le C_f$ and $0\le \tilde q_{n-1}(x)\le C_f$ for all $x$.

    Finally, since \cref{condition1}(b)--(c) holds for $\phi_1$, \cref{supp:proof:lemma4} applies to the score $S^{\phi_1}$ with $q_n=\tilde q_{n-1}$ and $q=f_k(\cdot;\theta^{*(2)}_k)$, and therefore implies
    \begin{equation}
        \E \Big| S^{\phi_1}(\tilde q_{n-1},X_n)-S^{\phi_1}(q,X_n)\Big| \to 0.
    \end{equation}
    This completes the proof of (a). 

    To prove (b), recall from \cref{sec:pre:scoring} that each proper score induces a divergence via $D(P,Q) := H(Q) - S(P,Q)$, where $H(Q) = S(Q,Q)$ is the generalized entropy, and $S(P,Q) = \E_{X\sim Q}S(P,X)$. For the Bregman score $S^{\phi_1}$, the induced divergence coincides with the separable Bregman divergence $D^{\phi_1}$. Therefore, for each fixed $k$,
    \begin{equation}
        D^{\phi_1}\big( g \| f_k(\cdot; \theta^{*(2)}_k) \big) = H^{\phi_1}(g) - \E\Big( S^{\phi_1}\big( f_k(\cdot;\theta^{*(2)}_k), X\big) \Big) = H^{\phi_1}(g) - \bar{u}^{(1|2)}_{k},
    \end{equation}
    where $H^{\phi_1}(g)$ does not depend on $k$. Hence maximizing $\bar{u}^{(1|2)}_{k}$ over $k$ is equivalent to minimizing
    $D^{\phi_1}(g \| f_k(\cdot;\theta^{*(2)}_k))$ over $k$, proving (b).
\end{proof}

\newpage
\section{Additional Figures}

\begin{figure}
    \centering
    \includegraphics[width=1\linewidth]{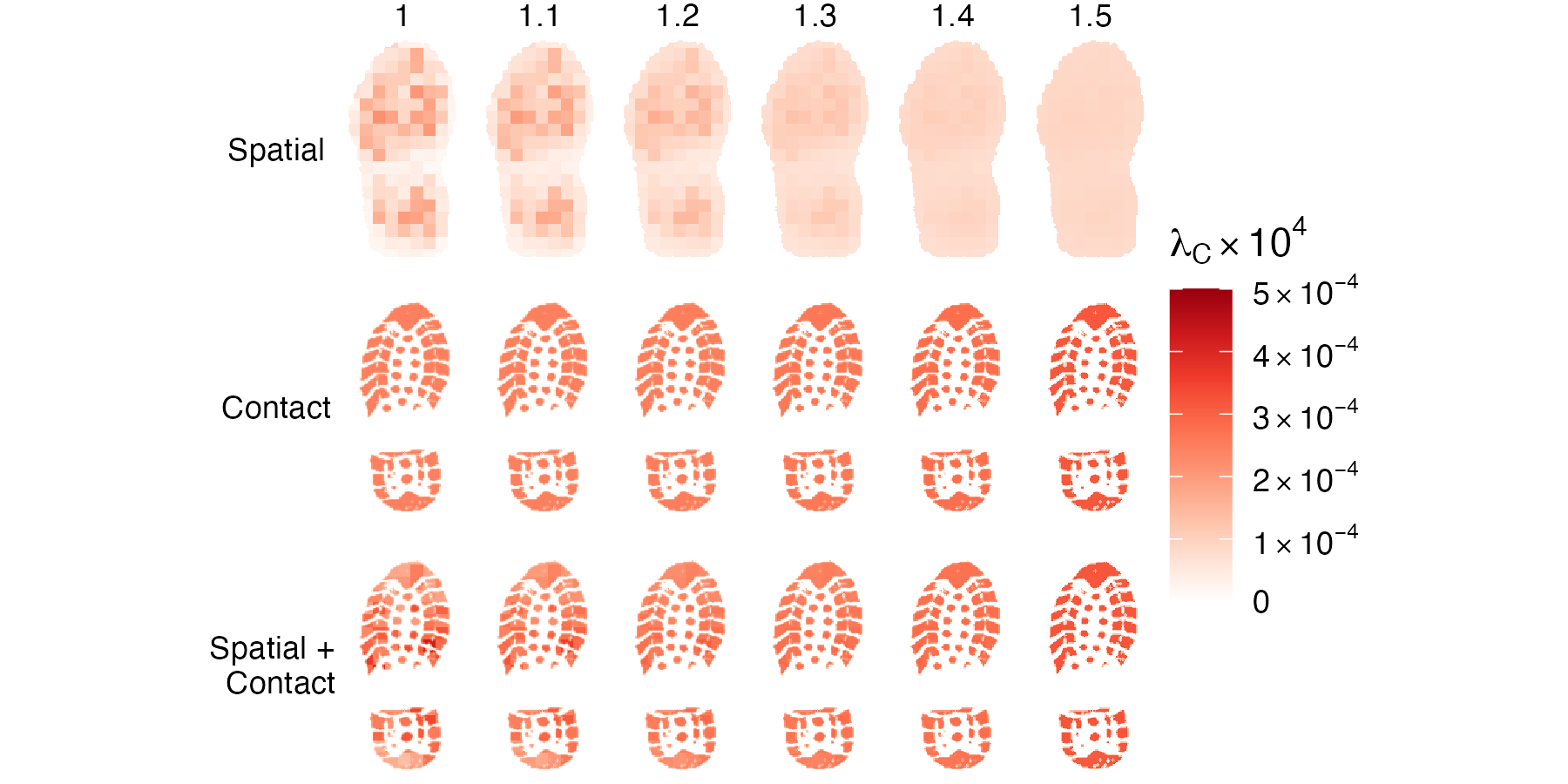}
    \captionsetup{font=small}
    \caption{A companion to \cref{fig:shoe_main_grid}. The posterior mean of $\lambda_{\text{C}}$ across different values of $\beta$ according to different models. Note that the Uniform Model is excluded as its fit does not depend on $\beta$. As $\beta$ increases, the Spatial Model fit approaches that of the Uniform Model, and the Spatial + Contact Model fit approaches that of the Contact Model.}
    \label{fig:shoe_predictives3}
\end{figure}

\putbib
\end{bibunit}
\end{document}